\documentclass[english,twoside,11pt]{article}%

\usepackage{jmlr2e}

\usepackage{algorithm, algorithmicx}
\usepackage{algpseudocode}
\usepackage{array}
\usepackage{amssymb}
\usepackage{booktabs}
\usepackage{graphicx}
\usepackage{mathtools}
\usepackage{subfiles}
\usepackage{todonotes}

\graphicspath{{./graphics/},{./figures/}}

\newcommand{\defeq}{\coloneqq}
\newcommand{\ud}{\,\mathrm{d}}
\newcommand{\mathd}{\mathrm{d}}
\newcommand{\tmmathbf}[1]{\ensuremath{\boldsymbol{#1}}}

\def\NN{{\mathbb N}}        
\def\RR{{\mathbb R}}        
\def\1{{\mathbf 1}}        

\newtheorem{assumption}{Assumption}

\DeclareMathOperator{\support}{support}

\newcommand\bracearraycond[1]{\left\{ \begin{array}{ll} #1 \end{array} \right.}
\newcommand\labelledmapsto[1]{\stackrel{#1}{\longmapsto}}
\newcommand\restr[2]{{
  \left.\kern-\nulldelimiterspace 
  #1 
  \vphantom{\big|} 
  \right|_{#2} 
}}

\newcolumntype{C}[1]{>{\centering\let\newline\\\arraybackslash\hspace{0pt}}m{#1}}

\usepackage{notation}

\let\rawpi\pi
\let\rawsigma\sigma
\let\rawphi\phi
\let\rawgamma\gamma
\let\rawtau\tau
\let\rawtheta\theta

\notation*[Index of a block (cluster)]{k}{k}
\notation*[Observation index]{i}{i}
\notation[Particle index]{p}{p}
\notation[Number of particles used in the particle Gibbs algorithms]{N}{N}
\notation*[SMC generation (iteration) index]{t}{t}
\notation[Number of points to allocate, also coincides with the number of particle Gibbs generations (iterations)]{n}{n}

\notation*[Bijection from particles to restricted clusterings]{phi}{\rawphi}
\notation[Proposal distribution for the particle Gibbs algorithm]{q}{q}
\notation*[Intermediate un-normalized distributions]{gamma}{\rawgamma}
\notation[Number of datapoints]{T}{T}
\notation*[One block in a clustering]{b}{b}
\notation[All observations]{yVec}{\bold{y}}
\notation[One observed datapoint]{y}{y}
\notation*[Partition or clustering of the observation indices]{c}{c}
\notation[Subset of observation indices; in particular, set of anchor points]{s}{s}
\notation*[Parameter of the parametric family]{theta}{\rawtheta}
\notation*[Likelihood of the parametric family]{L}{L}
\notation*[Prior for the parametric family, and base measure of the process]{H}{H}
\notation*[Prior over clusterings]{tau}{\rawtau}
\notation*[Target density of interest]{pi}{\rawpi}
\notation[Factor in the clustering prior depending on the number of clusters]{tauI}{\protect{\rawtau_1}}
\notation[Factor in the clustering prior depending on the size of each block]{tauII}{\protect{\rawtau_2}}
\notation[Closure of a subset of datapoint indices relative to a clustering]{sBar}{\overline{s}}
\notation[Restriction of the clustering to the closure subset]{cBar}{\overline{c}}
\notation[Split merge target distribution]{piBar}{\overline{\rawpi}}
\notation[Split merge modification of the prior factor depending on the number of clusters]{tauIBar}{\overline{\rawtau}_1}
\notation[A binary relationship indicating whether two datapoint indices are in the same block (relative to some clustering specified as a subscript)]{isclusteredwith}{\sim}
\notation[A support restriction on the split merge target]{I}{I}
\notation[Blocks of the partition outside of the closure]{cMinus}{c_{-}}
\notation[Auxiliary target distribution]{piTilde}{\tilde{\rawpi}}
\notation[Distribution of the anchors]{h}{h}
\notation[Markov kernel used to conceptualized a whole split-merge step]{K}{K}
\notation[Permutation auxiliary variable vector]{sigmaVec}{\boldsymbol{\rawsigma}}
\notation*[Permutation auxiliary variable component]{sigma}{\rawsigma}
\notation[Allocation decision (a particle is made of a vector of such allocations)]{x}{x}
\notation[State space containing the possible allocation decisions]{X}{\mathcal{X}}
\notation[Particle]{xVec}{\bold{x}}
\notation[Support defining valid sequences of allocation decisions]{partSupport}{\mathcal{S}}
\notation[Resampling ancestry variable]{ancestor}{a}
\notation[Resampling ancestry vector]{ancestors}{\boldsymbol{a}}
\notation*[Conditional resampling distribution]{r}{r}
\notation[Weight variable]{w}{w}
\notation[Weight vector]{wVec}{\boldsymbol{w}}
\notation[Markov kernel used to conceptualized a particle Gibbs step]{KBar}{\overline{K}}
\notation[Particle Gibbs auxiliary variables]{g}{g} 
\notation[Annealing schedule used to delay the introduction of the clustering prior]{schedule}{\zeta}
\notation[Sufficient statistics for the likelihood]{suffstat}{\psi}
\notation{nclust}{\protect{|c|}}
\notation[Relative ESS threshold used for adaptive resampling]{ressThreshold}{\beta}

\notation[Concentration parameter used in clustering priors such as the Dirichlet process]{concentration}{\protect{\alpha_0}}
\notation[Discount parameter used in the Pitman-Yor process]{discount}{d}
\notation[Maximum number of clusters in the finite Dirichlet model]{maxnclust}{\k_0}
\notation[Symmetric concentration parameters for the finite Dirichlet model]{finiteconcentration}{\alpha_{\textrm{Dir}}}

\makenomenclature

\jmlrheading{18}{2017}{1-39}{7/15; Revised 2/17}{4/17}{15-397}{Bouchard-C\^{o}t\'{e}, Doucet and Roth}

\ShortHeadings{Particle Gibbs Split-Merge Sampling}{Bouchard-C\^ot\'e, Doucet and Roth}

\firstpageno{1}

\begin{document}

\title{Particle Gibbs Split-Merge Sampling for Bayesian Inference in Mixture Models}

\author{%
	\name Alexandre Bouchard-C\^ot\'e \email bouchard@stat.ubc.ca\\
	\addr Department of Statistics \\
	University of British Columbia \\
	Corresponding address: 3182 Earth Sciences Building, 2207 Main Mall, Vancouver, BC, Canada V6T 1Z4
	\AND
	\name Arnaud Doucet \email doucet@stats.ox.ac.uk\\
	\addr Department of Statistics \\
	University of Oxford, United Kingdom
	\AND
	\name Andrew Roth \email andrew.roth@ludwig.ox.ac.uk \\
	\addr Department of Statistics and Ludwig Institute for Cancer Research\\
	University of Oxford, United Kingdom
}

\editor{Zhihua Zhang}

\maketitle

\begin{abstract}%
This paper presents an original Markov chain Monte Carlo method to sample from the posterior distribution of conjugate mixture models.
This algorithm relies on a flexible split-merge procedure built using the particle Gibbs sampler introduced in \cite{andrieudoucet2009,andrieu2010}.
The resulting so-called Particle Gibbs Split-Merge sampler does not require the computation of a complex acceptance ratio and can be implemented using existing sequential Monte Carlo libraries.
We investigate its performance experimentally on synthetic problems as well as on geolocation data.
Our results show that for a given computational budget, the Particle Gibbs Split-Merge sampler empirically outperforms existing split merge methods. 
The code and instructions allowing to reproduce the experiments are available at \url{https://github.com/aroth85/pgsm}. 

\end{abstract}

\section{Introduction}

Mixture models are very commonly used to perform clustering and density estimation, and they have consequently found numerous applications in a wide range of scientific fields.
Since the introduction of Markov chain Monte Carlo (MCMC) methods in statistics over twenty five years ago, the Bayesian approach to mixture models has become very popular \citep{marin2005,richardsongreen1997}.
However, sampling from the posterior distribution of mixture models remains a challenging computational problem.

When conjugate priors are used, it is possible to analytically integrate out the mixing proportions and the parameters of the components.
This is the scenario we will focus on in this article.
In this case, we aim to sample from the posterior distribution of the latent  indicator variables associated with the observations, each latent variable indicating which component of the mixture generates a given data point.
A simple Gibbs sampler can be used which updates the latent indicator variables one-at-a-time \citep{mceachern1994,escobarWest1995} but this algorithm is inefficient when the number of observed data points $\T$ is large.
First, the simulated Markov chain would have to visit a long chain of lower probability configurations in order to split and merge large clusters.
As a result, it is prone to getting trapped in severe local modes.
Second, it is non-trivial to parallelize due to the inherently sequential nature of the updates.

The limitations of the simple Gibbs sampler has motivated a rich literature on MCMC algorithms for Bayesian  mixture models which partially address these issues; see, e.g., \cite{Ishwaran2001,Liang2007,walker2007sampling,kalli2011slice}. In particular, procedures proposing to split and merge existing clusters in one single step have
become prominent as they generally perform better than the simple Gibbs sampler \citep{richardsongreen1997,neal2000,dahl2005,jain2004}.

While designing an efficient merge proposal is simple, designing an efficient split proposal is a more complicated task.
When the mixing proportions and parameters are not integrated out, split-merge moves were first proposed in \cite{richardsongreen1997}.
The proposals were built to ensure the conservation of some moments and accepted/rejected using Metropolis-Hastings steps.
However, it is difficult to design efficient proposals in this context.

When the mixing proportions and parameters are integrated out, split-merge moves on the latent indicator variables were first proposed in \cite{jain2004}.
Assume one is interested in splitting a block/cluster of points $\b\subset\{1,\dots,\T\}$ into two blocks.
We select two points in $\b$, which will be in distinct blocks after the split.
There are  $2^{|\b|-2}$ possible ways to split the original block $\b$, hence any efficient proposal needs to be informed by the observations corresponding to the indices in $\b$.
In \cite{jain2004}, one selects two points at random which are used as anchors.
When the two anchors are in separate clusters, a merging of the two clusters is proposed.
When the two anchors are in the same cluster, a split is proposed as follows: first, the two anchor points seed a pair of new clusters, and second, several restricted Gibbs scans are performed to reallocate the remaining points originally clustered with the anchors to the two new clusters.
All clusters which do not contain the anchors are not altered, leading to a restricted Gibbs move.
After either a split or a merge is proposed, the Metropolis-Hastings ratio is computed to accept or reject the move.
The number of Gibbs scans in the split move is a free tuning parameter for this sampler. In \cite{dahl2005}, an alternative approach is proposed for split moves.
The restricted Gibbs scans are replaced by a sequential allocation step whereby the anchors define two new clusters and all points which were originally clustered with these points are sequentially allocated to one of the anchor clusters.

These split-merge algorithms have become popular as they provide state-of-the-art performance but they are relatively difficult to implement due to their complex Metropolis-Hastings acceptance ratios.

In the present work, we propose a novel split-merge sampler based on the conditional Sequential Monte Carlo (SMC) algorithm appearing in the Particle Gibbs (PG) sampler \citep{andrieudoucet2009,andrieu2010}, which we call the Particle Gibbs Split Merge (PGSM) sampler.
Most of the complexity inherent to split-merge operators is encapsulated into the well-understood PG sampling procedure \citep{ChopinSingh2012}, and no acceptance ratio needs to be computed.
Moreover, as the PGSM sampler relies on SMC methods, it benefits from advanced simulation methods from the SMC literature, such as adaptation schemes \citep{Lee2011} and methods for parallel and distributed inference \citep{Lee2010GraphicCards,Jun2012Entangled,Lee2014Forest}, as well as from efficient SMC software libraries \citep{Johansen2009smctc,Murray2013}.
The PGSM sampler does not make any topological assumption on the observation space in contrast to the posterior simulation techniques described in \cite{Dahl2003b} and \cite{Liang2007}.
This methodology complements the maximum \emph{a posteriori} inference techniques developed in \cite{Daume2007,Lianming2011}.

There has been previous work on applying sequential importance sampling and SMC methods for posterior simulation of Dirichlet processes and related mixture models.
However, to the best of our knowledge, SMC methods have never been previously used to design split-merge moves.
Indeed, the methods proposed in \cite{MacEACHERN1999,fearnhead2004dp,fearnhead2007dp2,Mansinghka2007,Caron2009Decomposable,carvalho2010} directly apply a single pass SMC algorithm to the entire clustering problem.
Empirical results in \cite{Kantas2015} suggest that such methods may require a number of particle which scales at least quadratically with respect to the number of datapoints.
The work of \cite{UlkerGC10} uses SMC within the context of the SMC Samplers methodology \citep{delmoral2006}, which makes it closer in spirit to existing MCMC methods.
Our contribution is to provide a principled approach for breaking down the clustering problem into smaller sub-problems more amenable to the use of SMC techniques.

Finally, other lines of work are devoted to parallelization and distribution of MCMC methods for mixture models \citep{chang13dpmm,williamson2013parallel,icml2014c2_gal14,ge2015distributed}.
As alluded to earlier, our method can potentially be parallelized and distributed using existing approaches from the SMC literature \citep{Lee2010GraphicCards,Jun2012Entangled,Lee2014Forest}. Like the other available split-merge procedures, it is also possible to consider different
split-merge moves simultaneously when the prior clustering distribution restricted to the clusters being updated does not depend on the number of clusters for the whole dataset. However, we do not focus on these aspects here.

The rest of this article is organized as follows.
Section~\ref{sec:Bayesianmixturemodels} introduces our notation for the types of Bayesian mixture models that we consider.
Section~\ref{sec:PGSMsampler} details the PGSM sampler. Section~\ref{sec:applications} applies the method to synthetic datasets, as well as real data from a geolocation application.
We conclude with some directions for future work and discussion in Section~\ref{sec:discussion}.

\section{Mixture models and Bayesian inference\label{sec:Bayesianmixturemodels}}\label{sec:mixture}

In this section we first layout notation and then describe Bayesian mixture models.
We focus on the case where the component base measure is conjugate to the data likelihood, so that the posterior distribution of any clustering can be evaluated analytically up 
to a normalizing constant.

\subsection{Notation}\label{sec:conventions}

We use bold letters for (random) vectors, and normal fonts for (random) scalars, sets, and matrices.
For quantities such as an individual observation $\y_i$, or a parameter $\theta$, which can be either scalars or vectors without affecting our methodology, we consider them as scalars without loss of generality.
Given a vector $\bold{x} = (x_1, x_2, \dots, x_n)$, and $i \le j$, we use $\bold{x}_{i:j}$ to denote the sub-vector $\bold{x}_{i:j} = (x_i, x_{i+1}, \dots, x_j)$.
To simplify notation, we do not distinguish random variables from their realization.
We define discrete probability distributions with their probability mass functions, and continuous probability distributions with their density functions with respect to the Lebesgue measure. 
A list of symbols is available in the Appendix.

\subsection{Bayesian mixture model}\label{sec:bayesian-mix-subsection}

Consider $\T$ observations $\yVec \defeq \left(\y_{1},\ldots,\y_{\T}\right)$.
A mixture model assumes that the observations indices $[\T] \defeq \{1, \ldots,\T\}$ are partitioned into subsets.
This partition is called a clustering, $\c \defeq \{\b_1,\ldots,\b_\nclust : \b_k \subseteq [\T]\}$ where $\nclust$ denotes the cardinality of the set $\c$ and each block $\b$ in the partition is referred to as a cluster.
Given the clustering $\c$, we define the following likelihood for the data

\begin{equation}\label{eq:likelihoodpartition}
\pi\left(\yVec \vert \c \right) \defeq \prod_{\b\in\c} \L(\yVec_\b),
\end{equation}
where $\L(\yVec_\b)$ is the likelihood of the observations in cluster $\b$

\begin{equation}\label{eq:conjugacy}
\L(\yVec_\b) \defeq\int \left( \prod_{i\in\b} \L(\y_\i \vert \theta) \right)  \H (\ud\theta ).
\end{equation}
In this expression, $\L(\y_\i \vert \theta)$ is a probability density function parametrized by $\theta$ and $\H(\ud\theta)$ a prior measure over this parameter.

The clustering $\c$ is unknown and is viewed as a random variable.
Let $\tau(\c)$ denote its prior probability, defined over the space of partitions of $[\T]$ and assumed to factorize as

\begin{equation}
\tau(\c) \propto \tauI(\nclust) \prod_{\b \in \c}
\tauII(|\b|) ,  \label{eq:priorpartition}
\end{equation}
where $\tauI:\NN\rightarrow \RR^{+}$ and $\tauII:\NN\rightarrow \RR^{+}$ are arbitrary functions.

This assumption on the prior clustering distribution is not restrictive and includes several popular priors, such as:

\begin{description}
	\item[Dirichlet process prior] \citep{Ferguson1973} with parameter $\concentration>0$: $\tauI(j) \propto\concentration^{j}$ and $\tauII(j) \propto (j-1)!$.
	\item[Pitman-Yor process prior] \citep{Pitman1997,Ishwaran2003} with parameters $\concentration, \discount$ ($\concentration>-\discount,$ $0\leq \discount < 1$): $\tauI(j) \propto \prod_{j'=1}^{j} \left\{\concentration+\discount \left(j'-1\right) \right\}$ and $\tauII(j) \propto\Gamma(j-\discount)$, where $\Gamma(\cdot)$ is the Gamma function.
	\item[Finite Dirichlet mixture] with parameter $\finiteconcentration>0$ and $\maxnclust$ components and symmetric concentration $\left(\finiteconcentration,\ldots,\finiteconcentration\right)$: $\tauI(j) \propto \1\left[j \le \maxnclust \right]$ and $\tauII(j) \propto \Gamma\left(j+\finiteconcentration\right)$.
\end{description}

The likelihood (\ref{eq:likelihoodpartition})-(\ref{eq:conjugacy}) and prior (\ref{eq:priorpartition}) define the following target posterior distribution

\begin{equation}\label{eq:target}
\pi(\c) \defeq \pi(\c\mid\yVec) \propto \tau(\c) \prod_{\b\in\c} \L(\yVec_\b).
\end{equation}
Since we view the observations as fixed, we drop the dependency on $\yVec$ from
the notation throughout the paper.
We detail in the following sections an original approach to sample from this posterior distribution.

\section{Methodology\label{sec:PGSMsampler}}

\begin{figure}[t]
	\begin{center}
		\includegraphics[scale=0.75]{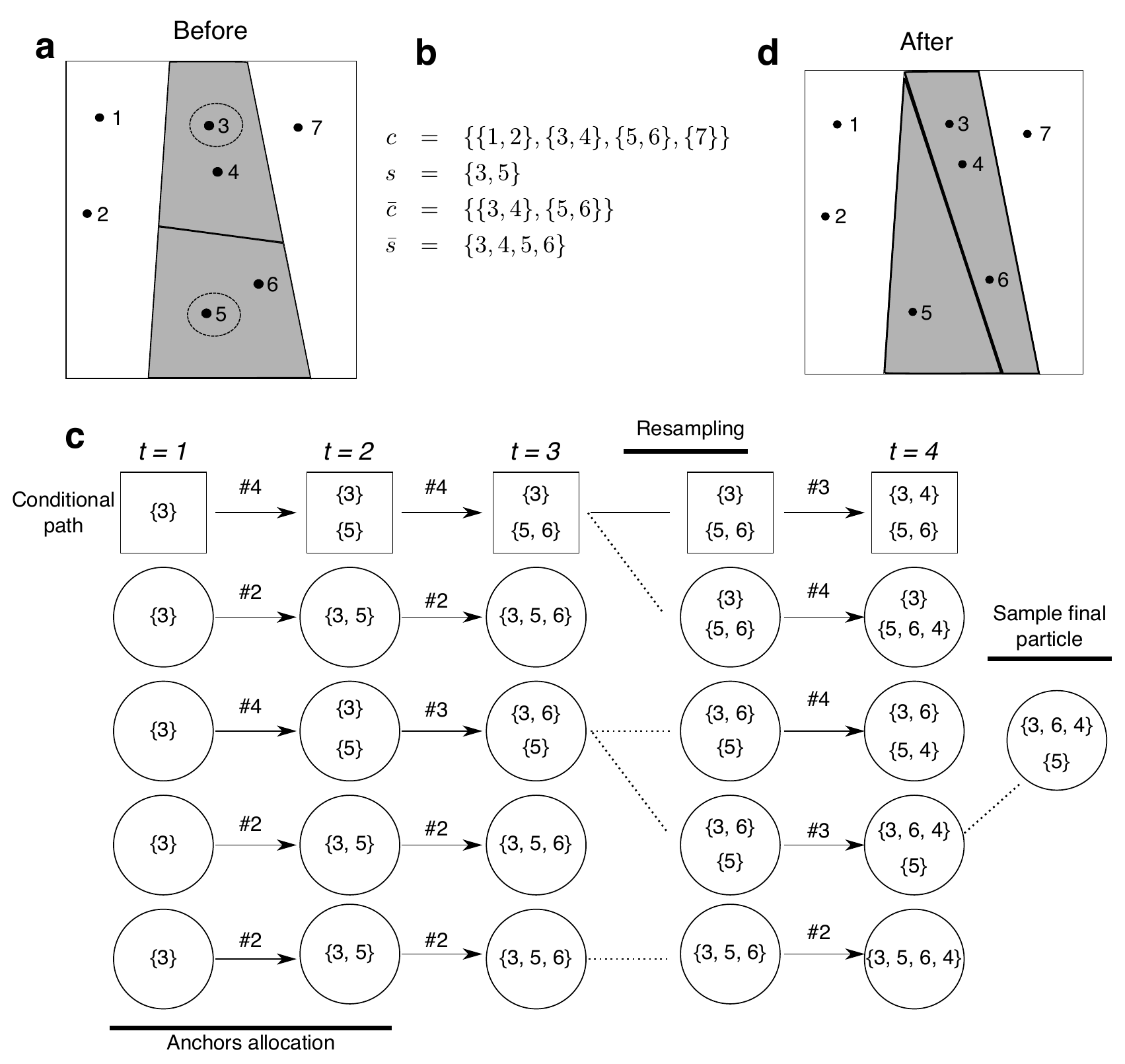}
	\end{center}
	\caption{
		(\textbf{a}) Example illustrating the setup of a split merge and (\textbf{b}) the notation used in Section~\ref{sec:decomposing}.
		The seven points denote the observation indices.
		While we show points in a two-dimensional space for visualization, the methodology does not rely on any topological properties of the observation space.
		The dashed circles denote the anchors $s$, and the shaded region, the closure of the anchors.
		(\textbf{c}) Example of a single PGSM\ step with $\boldsymbol{\rawsigma} = (3, 5, 6, 4)$. 
		Square boxes represent particles corresponding to the fixed conditional path.
		Circles represent regular particles.
		Each particle has an associated clustering denoted by the set or pair of sets written inside each particle.
		Arrows indicate proposal draws labelled by the proposed state (see Figure~\ref{fig:local-state-space}).
		Dashed lines indicate when resampling occurs.
		(\textbf{d}) Configuration after the PGSM update illustrated in \textbf{c}.
	}
	\label{fig:example}
\end{figure}

We organize the description of our method into two main parts.
First, we define a generic construction for decomposing the problem of sampling from the posterior (Equation~(\ref{eq:target})) with arbitrary numbers of clusters into \emph{split-merge} sub-problems.
Second, we show how the PG methodology can be used to address these sub-problems.

\subsection{Decomposing the clustering problem into split-merge subproblems}\label{sec:decomposing}

Algorithm~\ref{alg:setup_split_merge} allows us to break down the problem of sampling from the posterior into split-merge subproblems. We refer the reader to Figure~\ref{fig:example} for an illustrative example of the notation used throughout this description.

\begin{algorithm}
	\caption{}
	\begin{algorithmic}[1]
		\Function{SplitMerge}{$c, h(s), K_{c,s}(\bar c'|\bar c)$} 
		
		\State $\s \sim \h(\cdot)$ \Comment{$\s=\{i_{1}, i_{2}\}$}
		
		\State $\cBar \gets \{\b\in\c : \b \cap \sBar \neq \emptyset\}$  \Comment{Clustering restricted to the anchors}
		
		\State $\sBar \gets \bigcup_{\b \in \cBar} \b$ \Comment{Closure of the anchors with respect to the clustering $\cBar$}
		
		\State $\cBar' \sim \K_{\c,\s}(\cdot|\cBar)$
		
		\State $\c' \gets \cBar' \cup (\c \backslash \cBar)$ \label{step:create-final-cluster}
		
		\State \textbf{return} $\c'$
		
		\EndFunction
	\end{algorithmic}
	\label{alg:setup_split_merge}
\end{algorithm}

The algorithm requires three inputs:
\begin{enumerate}
	\item $\c$: the current clustering,
	
	\item $\h(\s)$: a distribution for proposing an unordered pair of \emph{anchors} $\s=\{i_{1}, i_{2}\} \subset [\T]$,
	
	\item $\K_{c,s}(\cBar'|\cBar)$: a Markov transition kernel over the space of partitions of $\sBar$.
	This kernel is assumed to be invariant with respect to the following target distribution:
	\begin{eqnarray}
	\piBar_{c,s}(\cBar') &\propto& \tauIBar(|\cBar'|) \left( \prod_{\b\in\cBar'} \tauII(|\b|) \L(\yVec_\b) \1\left[b \cap \s \ne \emptyset \right] \right), \label{eq:piBar} \\
	\tauIBar(j) &\defeq & \tauI(j + \nclust - |\cBar|). \nonumber
	\end{eqnarray}
\end{enumerate}
In the following, we drop the subscripts from the kernel $\K$ and target $\piBar$ for simplicity.

The distribution $\piBar$ has a form similar to the posterior distribution defined in Equation~(\ref{eq:target}) with two modifications.
First, $\tauI(\nclust)$ is replaced by $\tauIBar(|\cBar'|) $.
Second, the support of the distribution is restricted so that each block in $\cBar'$ must contain at least one anchor point.
This also implicitly enforces the constraint that $|\cBar'| \le |s| = 2$.

Algorithm \ref{alg:setup_split_merge} returns an updated clustering where only the allocation of points in $\sBar$ have changed; i.e. the updated clustering $\c'$ only potentially differs from $\c$ at points which were initially clustered with the anchor points. 
The anchor proposal distribution $\h$  obviously impacts the performance of this procedure.
We empirically compare the performance of three anchor proposal distributions $\h$ in Section~\ref{sec:artificial}.

This scheme has the following property:

\begin{proposition}\label{prop:split-merge-correctness}
	If $\c \sim \pi$, where $\pi$ is given by Equation~(\ref{eq:target}), then the output of Algorithm~\ref{alg:setup_split_merge}, $\c'$, satisfies $\c' \sim \pi$.
	That is, the Markov kernel $\K(\c'|\c)$ induced by Algorithm~\ref{alg:setup_split_merge} is $\pi$-invariant.
\end{proposition}

\subsection{Overview of the particle Gibbs algorithm}\label{sec:pgsm-overview}

Ideally, we would like to sample independently from $\piBar$ in Algorithm  \ref{alg:setup_split_merge}, that is we would like to have $K(\cBar'|\cBar)=\piBar(\cBar')$, but this is too computationally expensive if $|\sBar|$ is large. Our primary contribution is an original way to address this issue using SMC-based methods.

In principle, it would be possible to use SMC methods to obtain a sample approximately distributed according to $\piBar$ \citep{fearnhead2004dp}. However, if we were to use this sample within Algorithm~\ref{alg:setup_split_merge}, the resulting invariant distribution would not be $\pi$. For this reason, we consider Particle MCMC\ (PMCMC) methods \citep{andrieudoucet2009,andrieu2010}.

PMCMC methods allow us to use SMC\ ideas in a principled way within MCMC\ schemes.
We will focus here on the PG sampler and show how one can use this methodology to obtain an efficient MCMC\ kernel  $\K$ targeting the distribution $\piBar$ given in Equation~(\ref{eq:piBar}). The outcome of the PG sampling steps will be either to cluster all the points in $\sBar$ into one block or to break $\sBar$ into two clusters, with the restriction that each of the two blocks should contain one anchor.
Interestingly, the form of the PG\ algorithm is the same no matter if the two anchors were previously together or apart before its execution.
This contrasts with previous split-merge algorithms such as \cite{jain2004}, which require a different treatment for split and merge moves.

To sample from $\piBar$, PG breaks the sampling of $\cBar'$ into a sequence of $\n \defeq |\sBar|$ simpler sampling problems.
In this scenario, contrary to most applications of PG, there is no intrinsic time ordering of the observations.
We randomize the order in which the points are included by introducing, conditionally on $\s$ and $\sBar$, a random permutation $\sigmaVec \defeq (\sigma_1, \dots, \sigma_\n)$.
This permutation is sampled using Algorithm~\ref{alg:sample_permmutation}.
\begin{algorithm}
	\caption{•}
	\begin{algorithmic}[1]
		\Function{SamplePermutation}{$\s, \sBar$}
		
		\State $\sigma_1 \sim \mbox{Uniform}(\s)$
		
		\State $\sigma_2 \gets \s\backslash\{\sigma_1\}$
		
		\State $(\sigma_3, \sigma_4, \dots, \sigma_\n) \gets \mbox{UniformPermutation}(\sBar \backslash \s)$
		
		\State \textbf{return} $\sigmaVec$
		
		\EndFunction
	\end{algorithmic}
	\label{alg:sample_permmutation}
\end{algorithm}

In other words, $\sigmaVec$ is uniform over the permutations of the observation indices in $\sBar$ such that the members of $\s$ appear in the first two entries.
The variable $\sigma_\t$ specifies the index of the observation $\y_{\sigma_\t}$  introduced into the PG algorithm at SMC\ iteration (``algorithmic'' time) $\t$ and $\x_\t$ is the corresponding allocation decision.
A particle $\xVec_{\t}$ is defined as a sequence of allocation decisions, $\xVec_{\t} \defeq (\x_1, \dots, \x_\t)$, where $\x_\t \in \X$, $\t \in \{1, \dots, \n\}$.

Given $\sigmaVec$, we denote the SMC\ proposals used within PG by $\q^\sigmaVec_\t(\x_\t|\xVec_{\t-1})$, and the intermediate unnormalized target distributions, by $\gamma^\sigmaVec_\t(\xVec_{\t})$.
We remind the reader that both $\gamma^\sigmaVec_\t$ and $\q^\sigmaVec_\t$ are allowed to depend on arbitrary subsets of the observations $\yVec$; see, e.g., \citep{delmoral2006}. 
However, we omit this dependency for notational simplicity.
Our methodology is flexible with respect to the choice of the proposals and the choice of the intermediate unnormalized target distributions.
For our methodology to provide consistent estimates, only the following weak assumptions have to be satisfied.

\begin{assumption}\label{assumption:support}
	For all $\t \in \{1, \dots, \n\}$, we assume $\support(\gamma^\sigmaVec_\t)  \subseteq  \support(\q^\sigmaVec_{\t})$ where  $\q^\sigmaVec_{\t}(\xVec_{\t}) \defeq \q^\sigmaVec_1(\x_1) \prod_{\k= 2}^{\t} \q^\sigmaVec_{\k}(\x_{\k} | \xVec_{\k-1})$ for $\t\geq 2$.
\end{assumption}

\begin{assumption}\label{assumption:bijection}
	We assume that there exists a bijection $\phi^\sigmaVec$ taking a particle as input, and outputting a clustering of $\sBar$.
	More precisely, $\phi^\sigmaVec$ is a bijection between the support of the proposal, and the support of the split-merge target distribution, $\phi^\sigmaVec : \support(\q^\sigmaVec_{\n}) \to \support(\piBar)$.
\end{assumption}

\begin{assumption}\label{assumption:intermediate-final}
	We assume that $\gamma^\sigmaVec_\n(\xVec_{\n}) \propto \piBar(\phi^\sigmaVec(\xVec_{\n}))$.
\end{assumption}

Assumption~\ref{assumption:support} ensures that all the importance weights appearing in the SMC method are well-defined.
Assumption~\ref{assumption:bijection} is a simple condition ensuring that we can consistently relabel the particles. 
Assumption~\ref{assumption:intermediate-final} ensures that we target the desired distribution at algorithmic time $n$.
Note that Assumption~\ref{assumption:intermediate-final} only restricts the choice of $\gamma_\t$ for the final SMC\ iteration, $\t = \n$. We use this flexibility in Section~\ref{sec:improved}.
We show in the next section how to design $\q^\sigmaVec$, $\gamma^\sigmaVec$, $\phi^\sigmaVec$ and $\X$ that satisfy these assumptions.

PG proceeds in a way similar to standard SMC\ algorithms, with the important difference that one of the $\N$ particle paths is fixed.
In our setup, this path is obtained using the inverse of the bijection described in Assumption~\ref{assumption:bijection}, applied to the state of the restricted clustering $\cBar$ prior to the current PG step.
As discussed in \cite{ChopinSingh2012}, we can without loss of generality set the genealogy of the conditioning path $\cBar$ to $(1,...,1)$, i.e. we use the particle index $\p = 1$ for this conditioning path: $\xVec_{\n}^1 \defeq \left(\phi^\sigmaVec\right)^{-1}(\cBar)$.
This defines a path by taking a prefix of length $\t$ of the vector $\xVec_{\n}^1$ for $\xVec_{\t}^1$, i.e. $\xVec_{\t}^1 = \left( \xVec_{\n}^1 \right)_{1:\t}$.

The final ingredient required to describe the PG algorithm is a \emph{conditional resampling} distribution $\r(\ancestors\mid\wVec)$, where $\ancestors \defeq (\ancestor_2, \dots, \ancestor_N)$ denotes the resampling ancestors, $\ancestor_\p \in \{1, \dots, \N\}$, and $\wVec \defeq (\w^1, \dots, \w^\N)$ denotes a vector of probabilities.
We limit ourselves to  multinomial resampling:
\begin{align}\label{eq:resampling}
\r(\ancestors\mid\wVec) = \prod_{\p=2}^\N r(\ancestor_{\p} \mid\wVec) &= \prod_{\p=2}^\N  \w^{\ancestor_\p}.
\end{align}%
More elaborate schemes can be used, see \cite{andrieudoucet2009,andrieu2010}. Instead of resampling at each time step as in vanilla SMC\ algorithms, we only resample when the relative Effective Sampling Size (ESS) criterion, which takes values between $0$ and $1$, is below a pre-specified threshold $\ressThreshold$, $\ressThreshold \in [0,1]$. The adaptive resampling procedure was proposed by \cite{Liu1995AdaptResampling} for standard particle methods and the correctness of this procedure for PG has been established in \cite{Lee2011}. The resulting procedure is described in Algorithm~\ref{alg:pgsm}.

\begin{algorithm}[h]
	\caption{•}
	\begin{algorithmic}[1]
		\Function{ParticleGibbsSplitMerge}{$\s, \sBar, \cBar$, $\piBar$} \Comment{Inputs coming from Algorithm~\ref{alg:setup_split_merge}}
		
		\State $\sigmaVec \gets \Call{SamplePermutation}{\s, \sBar}$
		
		\State $\xVec_{\n}^1 \gets \left(\phi^\sigmaVec\right)^{-1}(\cBar)$ \Comment{Compute the conditional path}
		
		\For{$\t \in \{1, \dots, \n-1\}$}
		\State $\xVec_{\t}^1 \gets \left( \xVec_{\n}^1 \right)_{1:\t}$ \Comment{First particle of each generation matches the conditional path}
		\EndFor
		
		\For{$p \in \{2, \dots, N\}$} \Comment{Initialize particles}
		\State $\x_1^\p \sim \q_1^\sigmaVec\left(\cdot\right)$
		
		\State $\xVec_1^\p \gets (\x_1^\p)$
		
		\EndFor
		
		\For{$p \in \{1, \dots, N\}$} \Comment{Initialize incremental importance weights}
		\State $\tilde \w_1^\p \gets \frac{\gamma^\sigmaVec_1(\xVec_1^\p)}{\q^\sigmaVec_1(\xVec_1^\p)}$
		
		\EndFor
		
		\For{$p \in \{1, \dots, N\}$}
		\State $\w_1^\p \gets \frac{\tilde \w_1^\p}{({\sum_{\p'=1}^\N \tilde \w_1^{\p'}})}$ \Comment{Compute normalized weights}  
		
		\EndFor
		
		\For{$t \in \{2, \dots, \n\}$}
		
		\If{$(\N \sum_{\p=1}^{\N} (\w_{t-1}^p)^2)^{-1} < \ressThreshold$} \Comment{Resample only if relative ESS is too low}
		
		\State $\ancestors \sim \r(\cdot\mid\wVec_{\t-1})$ \Comment{Perform the conditional resampling step} \label{alg:line:resampling}
		
		\State $\tilde \wVec_{\t-1} \gets (1, 1, 1, \dots, 1)$ \Comment{Reset the weights}
		
		\Else
		
		\State $\ancestors \gets (1, 2, 3, \dots, \N)$ \Comment{Resampling  is skipped: set $\ancestors$ to the identity map}

		\EndIf
		
		\For{$p \in \{2, \dots, N\}$}
		\State $\x_\t^\p \sim \q_\t^\sigmaVec\left(\cdot \mid \xVec_{\t-1}^{\ancestor_\p}\right)$ \Comment{Propose new block allocation for $\y_{\sigma_t}$}
		
		\State $\xVec_\t^\p \gets (\xVec_{\t-1}^{\ancestor_\p}, \x_\t^\p)$ \Comment{Concatenate new block allocation to path}
		
		\EndFor
		
		\For{$p \in \{1, \dots, N\}$}
		\State $\tilde \w_\t^\p \gets \tilde \w_{\t-1}^\p \cdot \w(\xVec_{\t-1}^{\ancestor_\p}, \x_\t^\p)$ \Comment{Update weights (see Equation~(\ref{eq:particle-weights}))} \label{alg:line:incremental_weight}
		
		\EndFor
		
		\For{$p \in \{1, \dots, N\}$}
		\State $\w_\t^\p \gets \frac{\tilde \w_\t^\p}{({\sum_{\p'=1}^\N \tilde \w_\t^{\p'}})}$ \Comment{Compute normalized weights}  \label{alg:line:weight_norm}
		
		\EndFor

		\EndFor \label{step:end-of-aux-vars}
		
		\State  $\xVec'_\n \sim \sum_{\p = 1}^\N \w_\n^\p \delta_{\xVec_\n^\p}(\cdot)$ \Comment{Sample particle representing new state} \label{alg:line:sample_particlepath}
		
		\State $\cBar' \gets \phi^\sigmaVec(\xVec'_\n)$ \Comment{Compute updated partition} \label{step:update-part}
		
		\State \textbf{return} $\cBar'$
		
		\EndFunction
	\end{algorithmic}
	\label{alg:pgsm}
\end{algorithm}

Most of Algorithm~\ref{alg:pgsm} is concerned with the creation of temporary auxiliary variables (lines 1--\ref{step:end-of-aux-vars}).
These auxiliary variables can all be discarded after the algorithm returns $\cBar'$, as they can be resampled from scratch every time Algorithm~\ref{alg:pgsm} is run. 
The part of the algorithm that performs the actual split or merge is in lines \ref{alg:line:sample_particlepath} and \ref{step:update-part}.
At this point of the execution of the algorithm, the particle population at SMC\ generation $\n$ can be interpreted (via $\phi^\sigmaVec$) as a distribution over clusterings of $\sBar$, with some particles 
corresponding to merging all points in $\sBar$ into one block (i.e. when $|\phi^\sigmaVec(\xVec'_\n)| = 1$), and others, to various way of splitting $\sBar$ into two blocks (when $|\phi^\sigmaVec(\xVec'_\n)| = 2$).

Correctness of this procedure follows straightforwardly from the original PG argument (see Appendix~\ref{appendix:pg-correctness} for details):

\begin{proposition}\label{prop:pg-correctness}
	Under Assumptions~\ref{assumption:support}, \ref{assumption:bijection}, and \ref{assumption:intermediate-final}, the output of Algorithm~\ref{alg:pgsm}, $\cBar'$, satisfies $\cBar' \sim \piBar$  if $\cBar \sim \piBar$, for any $\N\geq2$, i.e. the Markov kernel $\KBar(\cBar'|\cBar)$ induced by Algorithm~\ref{alg:pgsm} is $\piBar$-invariant.
\end{proposition}

\subsection{Intermediate target distributions and proposals construction}\label{sec:intermediate}

We detail here the construction of a set of proposal distributions $\q^\sigmaVec$, unnormalized target distributions $\gamma^\sigmaVec$, and mappings $\phi^\sigmaVec$ satisfying Assumptions~\ref{assumption:support} to \ref{assumption:intermediate-final}.
We denote the space of possible allocation decisions at a given PG\ iteration by $\X$. Our construction is based on an encoding where the space $\X$ consists in the rectangles shown in Figure~\ref{fig:local-state-space}.
We call the rectangles \emph{states} for short.
These states are used to build particles: recall that a particle $\xVec_\t$ is defined as a list of local decisions, $\xVec_{\t} \defeq (\x_1, \dots, \x_\t)$, $\x_{\t'} \in \X$.

\begin{figure}[t]
	\begin{center}
		\begin{tabular}{cc}
			\includegraphics[width=3.5in]{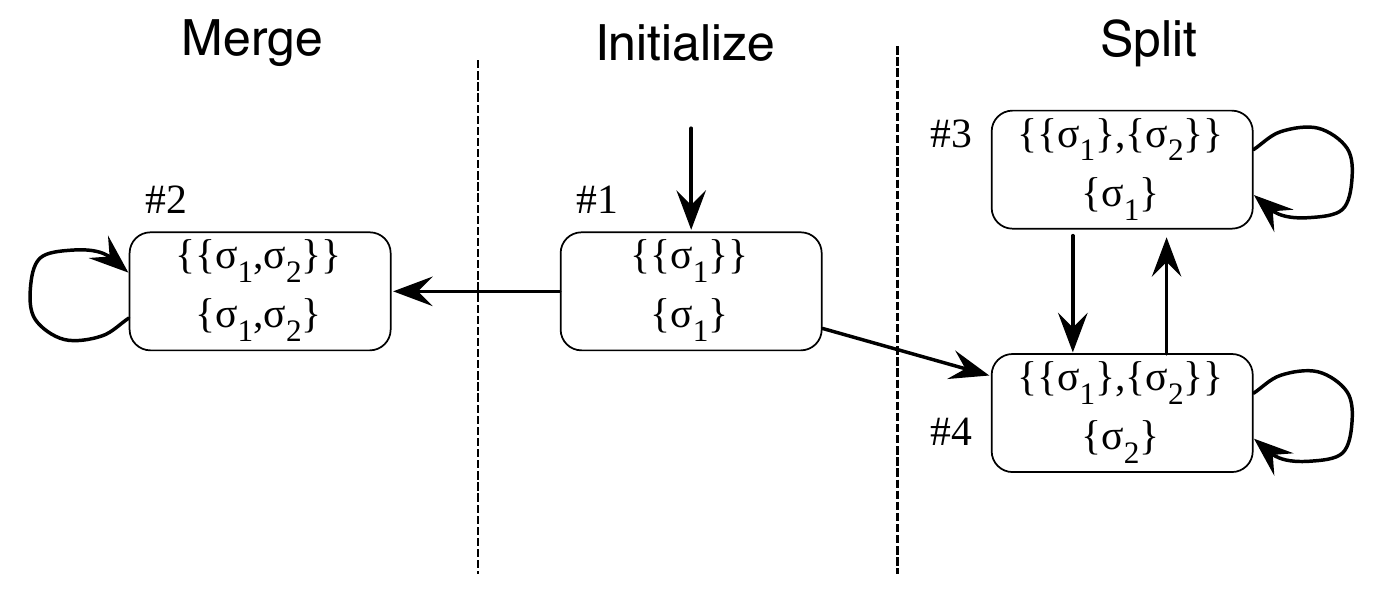} &
			\begin{minipage}{1.7in}
				\vspace{-1in}
				\begin{eqnarray*}
					\partSupport(\#1) &\defeq& \{\#2, \#4\} \\
					\partSupport(\#2) &\defeq& \{\#2\} \\
					\partSupport(\#3) &\defeq& \{\#3, \#4\}  \\
					\partSupport(\#4) &\defeq& \{\#3, \#4\}
				\end{eqnarray*}
				\\
				\\
			\end{minipage}
		\end{tabular}
		\caption{
			Left: State space $\X$ and allowed transitions $\partSupport(\cdot)$ for the local allocation decisions.
			Right: allowed transitions between the states.
		}
		\label{fig:local-state-space}
	\end{center}
\end{figure}

The state appended to a particle at time $\t$ represents (a) the clustering restricted to the anchors (shown in the first line of each rectangle in Figure~\ref{fig:local-state-space}), and (b), the cluster joined by $\y_{\sigma_\t}$ (encoded by the anchor(s) contained in the joined cluster, second line in the same figure).
As shown in Figure~\ref{fig:local-state-space}, the ``merge state'' (left) is an absorbing state, encoding the fact that following this local decision, all children particles are forced to join the unique block in the restricted clustering. The two ``split states'' (right), on the other hand, both have two outgoing transitions, encoding the fact that for each index in $\sBar \backslash \s$, the corresponding observation needs to be allocated to one of the two blocks.

There is a bijection between the support of $\piBar$, and particles respecting the transition constraints defined by the arrows in Figure~\ref{fig:local-state-space}.
More precisely, for each state $\x \in \X$, we let $\partSupport(\x)$ denote the set of allowed transitions from $\x$. We write $\xVec_{\t} \in \partSupport_\t$ if (a) $\x_1 = \#1$, and (b) for all $\t' \in \{2, \dots, \t\}$, $\x_{\t'} \in \partSupport(\x_{\t'-1})$.
From this definition, we obtain the following result whose proof is given in Appendix~\ref{appendix:bijection}.

\begin{proposition}\label{prop:bijection}
	For any permutation $\sigmaVec$ satisfying $\{\sigma_1, \sigma_2\} = \s$, there is a bijective map $\phi^\sigmaVec$ from the space of particles respecting the transition constraints, $\partSupport_\n$, to the support of the restricted target, $\support(\piBar)$.
\end{proposition}

We use this bijection to define a sequence of intermediate target and proposal distributions. The intermediate target at time $t$ of support $\partSupport_\t$ is given by:
\begin{equation}
\gamma^\sigmaVec_\t(\xVec_\t) \defeq \tauIBar(\cBar_\t) \left( \prod_{\b\in\cBar_\t} \tauII(|\b|) \L(\yVec_\b) \right),
\end{equation}
where $\cBar_\t = \phi^{\boldsymbol{\sigma}_{1:\t}}(\xVec_\t)$. 
By construction, we have that for $\t = \n$, $\gamma^\sigmaVec_\n(\xVec_{\n}) \propto \piBar(\phi^\sigmaVec(\xVec_{\n}))$ so Assumption~\ref{assumption:intermediate-final} is satisfied.

We define as proposals:

\begin{eqnarray}
\q^\sigmaVec_1(\x_1) &\defeq& \delta_{\#1}(\x_1), \\
\q^\sigmaVec_\t(\x_\t \mid \xVec_{\t-1}) &:=& \frac{\gamma^\sigmaVec_\t(\xVec_\t)}{\sum_{\x'_\t \in \partSupport(\x_{\t-1})}\gamma^\sigmaVec_\t(\xVec_{\t-1}, \x'_\t)}, \nonumber
\end{eqnarray}%
where $(\xVec_{\t-1}, \x'_\t)$ denotes the concatenation of $\x'_\t$ to the vector $\xVec_{\t-1}$, and $\xVec_{\t} = (\xVec_{\t-1}, \x_\t)$.
These definitions satisfy Assumption~\ref{assumption:support}, and yield the following weight updates:
\begin{eqnarray}\label{eq:particle-weights}
\w_\t(\xVec_{\t-1}, \x_t) &\defeq&  \frac{\gamma^\sigmaVec_\t(\xVec_\t)}{\gamma^\sigmaVec_{\t-1}( \xVec_{\t-1} )} \frac{1}{\q^\sigmaVec_\t(\x_\t \mid \xVec_{\t-1})} \\
&=& \frac{\sum_{\x'_\t \in \partSupport(\x_{\t-1})} \gamma^\sigmaVec_\t\left(\xVec_{\t-1}, \x'_\t\right)}{\gamma^\sigmaVec_{\t - 1}(\xVec_{\t-1})}  \nonumber \\
&=& \sum_{\x'_\t \in \partSupport(\x_{\t-1})} \frac{\gamma^\sigmaVec_\t\left(\xVec_{\t-1}, \x'_\t\right)}{\gamma^\sigmaVec_{\t - 1}(\xVec_{\t-1})}. \nonumber \label{eqn:gamma_ratio}
\end{eqnarray}
If $t > |\s|$ then Equation~(\ref{eqn:gamma_ratio}) simplifies as follows
\begin{eqnarray}
\frac{\gamma^\sigmaVec_\t(\xVec_\t)}{\gamma^\sigmaVec_{\t-1}(\xVec_{\t-1})} &=& \frac{\tauII(|\b_\t^+|)}{\tauII(|\b_\t^-|)} \L\left(\yVec_{\b_\t^+}\mid \yVec_{\b_\t^-}\right), \label{eq:tau_ratio}
\end{eqnarray}
where
\begin{equation}
\L\left(\yVec_{\b_\t^+}\mid \yVec_{\b_\t^-}\right) \defeq \frac{\L\left(\yVec_{\b_\t^+}\right)}{\L\left(\yVec_{\b_\t^-}\right)}\;.
\end{equation}
Here $\b_t^-$ and $\b_t^+$ encode the block in which a point is added to when transitioning from $\xVec_{\t-1}$ to $\xVec_\t$, the first being the block before the addition, and the second, the same block after the addition:
\begin{equation}
\b_t^- \defeq \cBar_{\t-1} \backslash \cBar_\t
,\;\; 
\b_t^+ \defeq \b^- \cup \{\sigma_\t\}	
.
\end{equation}
Depending on the form of the partition prior and likelihood it may be possible to simplify these quantities into more computationally efficient forms.

\subsection{An improved sequence of intermediate target distributions}\label{sec:improved}

We now describe an improvement over the basic intermediate and proposal distributions presented in the previous section.
This improvement addresses a ``greediness" problem of the (conditional) SMC procedure.
Consider a case where the ratio $\tauIBar(1)/\tauIBar(2)$ between a merge and a split is large.
This can occur for example when the Dirichlet process concentration parameter $\concentration$ is small.
In this case, the proposal in the first non-trivial step, $\q^\sigmaVec_2$, will assign most of its mass to the transition from state $\#1$ to state $\#2$ (see Figure~\ref{fig:local-state-space}).
However, the likelihood might overcome this prior when $|\sBar|$ is large.
But proposing such split has low probability under the definitions given in the previous section, as $\#2$ is an absorbing state.

To overcome this issue, we build a new sequence of intermediate distributions, which delay the incorporation of the prior:
\begin{equation}
\widehat{\gamma^\sigmaVec_\t}(\xVec_\t) \defeq \bracearraycond{
	\1[\xVec_\t \in \partSupport_\t],&\;\;\textrm{if }\t\in\{1,2\}, \\
	\left( \gamma^\sigmaVec_2(\xVec_{1:2}) \right)^{\schedule_\t} \frac{\gamma^\sigmaVec_\t(\xVec_\t)}{\gamma^\sigmaVec_2(\xVec_{1:2}),}&\;\;\textrm{otherwise.}
}
\end{equation}
where $\schedule_\t$ is a positive increasing annealing schedule such that $\schedule_\n = 1$.
We use the following proposal based on these new intermediate distributions:
\begin{eqnarray}
\widehat{\q^\sigmaVec_1}(\x_1) &\defeq& \delta_{\#1}(\x_1), \\
\widehat{\q^\sigmaVec_\t}(\x_\t \mid \xVec_{\t-1}) &:=& \frac{\widehat{\gamma^\sigmaVec_\t}(\xVec_\t)}{\sum_{\x'_\t \in \partSupport(\x_{\t-1})} \widehat{\gamma^\sigmaVec_\t}(\xVec_{\t-1}, \x'_{\t-1})}. \nonumber
\end{eqnarray}
This yields the weight updates:
\begin{eqnarray}\label{eq:final-particle-weight}
\widehat{\w_\t}(\xVec_{\t-1}, \x_t) &:=& \sum_{\x'_\t \in \partSupport(\x_{\t-1})}\frac{\widehat{\gamma^\sigmaVec_\t}\left( \xVec_{\t-1}, \x'_\t\right)}{\widehat{\gamma^\sigmaVec_{\t - 1}}(\xVec_{\t-1})}.
\end{eqnarray}
For simplicity, we pick $\schedule_\t = \frac{\t - 2}{\n - 2}$. This choice simplifies ratios of intermediate distributions to:
\begin{eqnarray}
\frac{\widehat{\gamma^\sigmaVec_\t}(\xVec_\t)}{\widehat{\gamma^\sigmaVec_{\t-1}}(\xVec_{\t-1})} &=& \bracearraycond{
	1&\;\;\textrm{if }\t = 2, \\
	\left( \gamma^\sigmaVec_2(\xVec_{1:2}) \right)^{\Delta \schedule} \frac{\gamma^\sigmaVec_\t(\xVec_\t)}{\gamma^\sigmaVec_{\t-1}(\xVec_{\t-1})}&\;\;\textrm{if }\t > 2,} \label{eq:weight-final} \\ \nonumber
\end{eqnarray}
where $\Delta \schedule\defeq (\n - 2)^{-1}$.

\subsection{Runtime analysis}\label{sec:runtime}

To simplify the analysis of the running time, we make a few assumptions.

\begin{assumption}\label{assumption:parametric-likelihood-running-time}
	The parametric likelihood model has the following properties:
	\begin{enumerate}
		\item let $\suffstat \defeq \suffstat(\yVec_\b)$ denote a sufficient statistic, and define, with a slight abuse of notation, $\L(\suffstat) \defeq \L(\yVec_\b)$.
		For a given sufficient statistic value $\suffstat$, the likelihood $L(\suffstat)$ can be computed in time $O(l)$,
		\item the sufficient statistic for $\b^+$, $\suffstat^+ \defeq \suffstat(\yVec_{\b^+})$, can be updated in time $O(u)$ from the sufficient statistic for $\b^-$, $\suffstat^- \defeq \suffstat(\yVec_{\b^-})$.
	\end{enumerate}
\end{assumption}
The next assumption holds for all the clustering priors reviewed in Section~\ref{sec:Bayesianmixturemodels}.
\begin{assumption}\label{assumption:ratio-running-time}
	The ratio $\frac{\tauII(j+1)}{\tauII(j)}$ can  be computed in constant time.
\end{assumption}
For example, with a Dirichlet process, this ratio is equal to $j!/(j-1)! = j$.  Since $|\b^+| = |\b^-| + 1$, Assumption~\ref{assumption:ratio-running-time} implies that the ratio $\frac{\tauII(|\b^+|)}{\tauII(|\b^-|)}$ in Equation~(\ref{eq:tau_ratio}) can be computed in constant time.
\begin{proposition}
	Under Assumptions~\ref{assumption:parametric-likelihood-running-time} and \ref{assumption:ratio-running-time}, one weight computation, Equation~(\ref{eq:final-particle-weight}), takes time $O(u+l)$.
	The storage cost per particle is $O(1)$. Moreover, the running time per weight computation is independent of the number of clusters.
\end{proposition}

The running time result follows directly from the fact that $|\partSupport(\cdot)| \le 2$, and hence the sum in Equation~(\ref{eq:final-particle-weight}) has a constant number of terms.
The constant storage cost follows from the finite dimensionality of the sufficient statistics (see Assumption~\ref{assumption:parametric-likelihood-running-time}), and from $|\X| = 4$.

We also remind the reader that for most resampling schemes, including the one in Equation~(\ref{eq:resampling}), the computational cost as a function of the number of particles and SMC\ iterations is $O(\N \n) = O(\N |\sBar|)$ \citep{Doucet2009Tutorial}.

\subsection{Generalization}

For simplicity, we have assumed so far that $|\s| = 2$, and hence, according to the auxiliary variable analysis of Appendix~\ref{appendix:pg-correctness}, $|\cBar| \le 2$.
In fact, the same auxiliary variables with more than two anchor points can be used to construct novel sampling algorithms. Details are given in Appendix~\ref{appendix:generalization}.

This generalization loses some interpretability compared to the split-merge case ($|\s| = 2$), but can be useful in finite clustering models.
In this case, it may only be possible to split a cluster if a merge is performed simultaneously.
For this reason, we use $|\s| = 3$ in the finite Dirichlet mixture model examples in Section~\ref{sec:artificial}.
For the Dirichlet Process, we did not observe notable improvements by going from $|\s| = 2$ to $|\s| = 3$, so we use the former setting for the non-parametric models.

\section{Applications\label{sec:applications}}

In this section, we demonstrate the performance of our methodology and compare it to standard alternatives.
We use a series of synthetic datasets covering a large spectrum of cluster separateness, as well as real data coming from a geolocation application.

\subsection{Implementation and evaluation}

We have implemented the following three Dirichlet Process (DP) clustering samplers in the same Python codebase: the PGSM method described in this work, the efficient Sequentially-Allocated Merge Split (SAMS) method of \cite{dahl2005}, as well as the standard Gibbs sampler. 
The code and instructions allowing to reproduce the experiments are available at \url{https://github.com/aroth85/pgsm}. 
The implementation of the likelihood computations are the same for all samplers, thus the running times are comparable.
We have tested the correctness of our computer implementations by computing the true posterior distribution on small examples via combinatorial enumeration, and verified that the Monte Carlo estimates converged to this distribution for all three methods.

Unless we state otherwise, we initialized the samplers with the single-cluster configuration.
In datasets much smaller than those studied in this work, initializing the Gibbs sampler to the fully disconnected clustering is advantageous \citep{sudderth06a}, but in larger datasets, the quadratic burn-in cost involved with this initialization is not scalable.
However, we verified that after a long burn-in period the Gibbs method initialized to the fully disconnected clustering eventually reaches the same likelihood values in the synthetic examples.
We also investigate the high cost of the fully disconnected initialization in the results shown in Figure~\ref{fig:mopsi}.

To evaluate the performance of the samplers, we held-out a random but fixed 10\% of each dataset.
We collected samples and computed the predictive likelihood and V-measure \citep{rosenberg2007vmeasure} every 100 iterations.
All experiments are replicated 10 times, and smoothed using a moving average with a window size of 20 for plotting.

\subsection{Likelihoods and priors}

In six of the synthetic experiments and the geolocation experiments discussed further in Section \ref{sec:artificial} and Section \ref{sec:geo}, we used a Normal-Inverse-Wishart conjugate likelihood model.
In two of the synthetic experiments in Section~\ref{sec:artificial} we used Bernoulli mixture models with 50 dimensions.
Each dimension is an independent draw from a Bernoulli random variable with cluster specific parameters.
We set a proportion of the dimensions to be uninformative as follows.
Values for uninformative dimensions were drawn from Bernoulli variables with parameter 0.5 regardless of the cluster membership.
Values for the remaining dimensions were drawn from cluster specific Bernoulli variables with parameters sampled from the Uniform distribution.
For the cancer data discussed in Section \ref{sec:genomics}, we use the application-specific PyClone likelihood model \citep{roth2014pyclone}.
The PyClone model uses genomic sequence data from tumours to identify mutations which co-occur in cells and estimates the proportion of cells harbouring the mutations.
The model is not conjugate, so we apply a discretization that allows us to treat the model as conjugate.
Complete details for each model are provided in Appendix~\ref{appendix:models}.

We use a DP prior with base measure given by the conjugate prior of the corresponding likelihood model in all experiments.
We use a $\mbox{Gamma}(1,0.1)$ prior and resample the value of the concentration parameters $\alpha_{0}$ using a standard auxiliary variable method \citep{escobarWest1995}.
The value of $\alpha_{0}$ is initialized to 1.0.

\subsection{Artificial datasets}\label{sec:artificial}

\begin{figure}[t]
	\centering
	\includegraphics[scale=1.0]{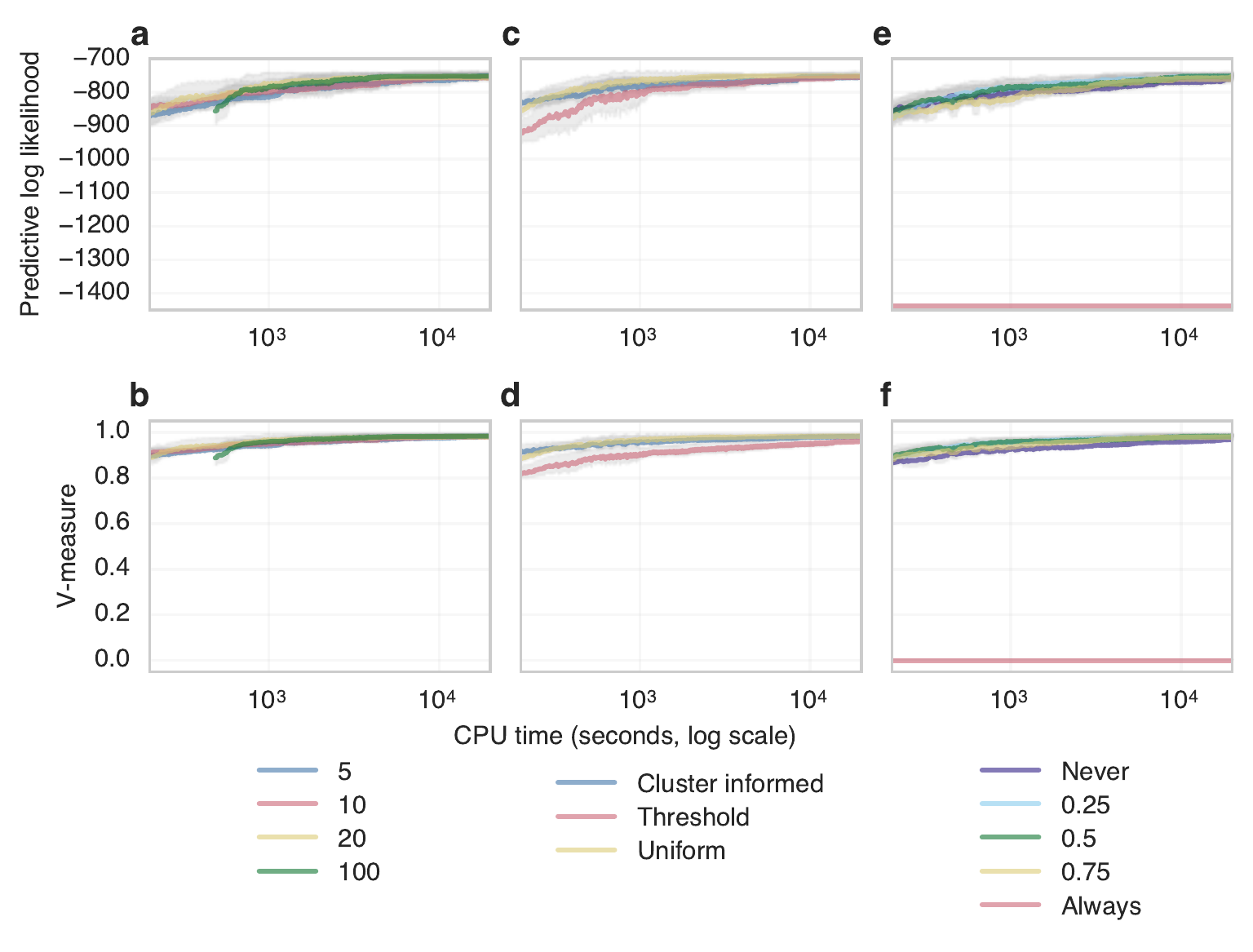}
	\caption{
		Effect on the predictive performance and clustering accuracy as a function of CPU time in log scale.
		\textbf{a}) and \textbf{b}) Varying the number of particles with the cluster informed proposal and a relative ESS resampling threshold of 0.5.
		\textbf{c}) and \textbf{d}) Varying the distribution for proposing pairs of anchor points, $h$, with 20 particles and an ESS resampling threshold of 0.5.
		\textbf{e}) and \textbf{f}) Varying the relative ESS resampling threshold with the cluster informed anchor proposal and 20 particles.
	}%
	\label{fig:pilot}%
\end{figure}

\begin{figure}[t]
	\centering
	\includegraphics{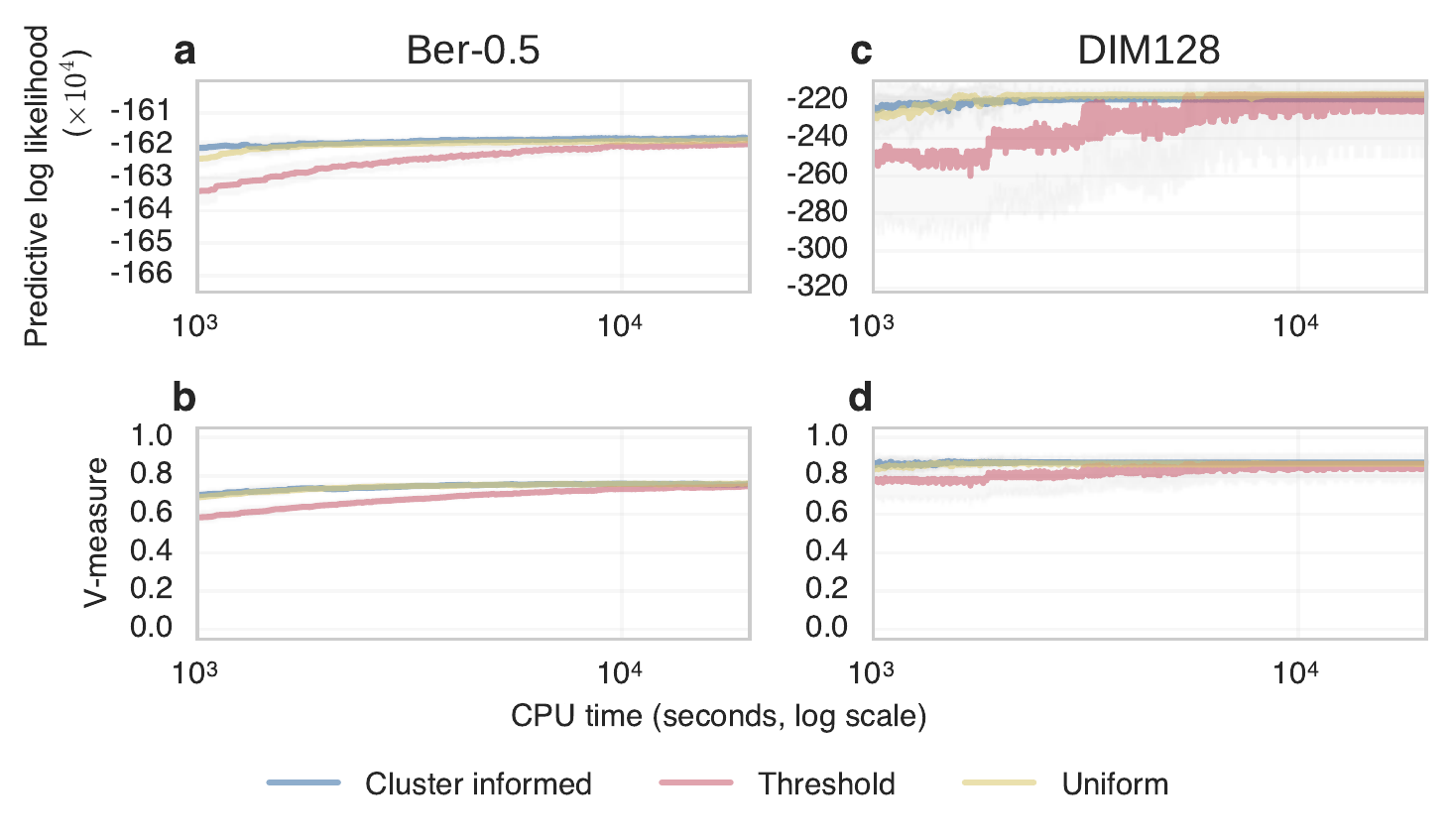}
	\caption{
		Effect on the predictive performance and clustering accuracy as a function of CPU time in log scale with different distribution $h$ for proposing pairs of anchor points.
		\textbf{a}) and \textbf{b}) Comparison using a 50 dimensional Bernoulli dataset with 50\% of the dimensions being uninformative.
		\textbf{c}) and \textbf{d}) Comparison using a 128 dimensional Normal dataset.
	}%
	\label{fig:pilot2}%
\end{figure}

We used four sources of synthetic data. 
First, the four datasets from \cite{Franti2006data} denoted S1--4.
Each of the four datasets consists in 5000 points generated from 15 bivariate Normal distributions with increasing amount of overlap between the clusters.
Second, we created another synthetic dataset, which we call C1, shown in Figure~\ref{fig:finite}, right.
Third, we simulated two datasets with 5000 points from a Bernoulli mixture model with 50 dimensions and 16 clusters, where we set 25\% (Ber-0.25) and 50\% (Ber-0.5) of dimensions to be uninformative.
Finally, we used 64 (DIM064) and 128 (DIM128) dimensional Normal datasets from \cite{dim_sets} with 1024 data points and 16 clusters.

We started with a series of pilot experiments on S1 only, designed to assess the effect of various tuning parameters on the performance of PGSM.
For all pilot experiments we use the unmixed PGSM sampler to isolate the effect of each tuning parameter.
In practice it is usually better to alternate between one iteration of the PGSM sampler and one iteration of the Gibbs sampler.
The effect of this mixing is explored later in this section.

We first explore the performance as we vary the number of particles used for each PGSM iteration (Figures~\ref{fig:pilot} \textbf{a} and \textbf{b}).
The curves with more particles take more time per iteration to run, however seem to achieve slightly better V-measure and predictive likelihood after the initial iterations.
The performance difference are negligible and the PGSM sampler generally seems insensitive to the number of particles for this dataset.
For subsequent experiments we used 20 particles.

Next we compare the performance of different proposal distributions for the anchor auxiliary variables (Figures~\ref{fig:pilot} \textbf{c} and \textbf{d}).
For this experiment we kept the number of particles fixed at 20 and the resampling threshold at 0.5.
We consider three proposal distributions.
\begin{description}
	\item[Uniform:] Sample the anchors uniformly at random from the $\binom{\T}{2}$ possibilities.
	\item[Cluster informed:] Sample the first anchor uniformly at random.
	Sample a cluster to draw the second anchor from with probability $\frac{1}{|c-1|}$ for the cluster containing the first anchor; otherwise proportional to  $$\frac{L(y_{\bar{b} \cup b})}{L(y_{\bar{b}}) L(y_{b})}$$ where $\bar{b}$ is the cluster containing the first anchor and $b$ is the candidate cluster.
	Sample the second anchor uniformly from the chosen cluster.
	\item[Threshold informed:] Sample the first anchor uniformly at random and the second anchor from clusters that have Chinese restaurant attachment probabilities greater than a threshold of 0.01.
\end{description}
To ensure the adaptation of the informed proposals stops, and does not perturb the invariant distribution of the sampler, we only update the proposal distributions when the number of clusters instantiated breaks the previous record.
Adaptation is guaranteed to terminate in finite time given there are only a finite number of points.
This would usually take a long time, so in practice it is advantageous to stop adaptation after a fixed period of time.
Detailed implementations of the cluster informed and threshold informed proposals are given in Appendix~\ref{appendix:proposals}.
Our results  suggest that performance is not strongly affected by the anchor proposal distribution $\h$.
We only saw a small advantage when using the informed proposal distributions for the auxiliary anchor variables in $\s$.
We also explored the effect of $h$ in higher dimensional datasets (Figure~\ref{fig:pilot2}).
Again we found the results are not sensitive to the choice of $h$.
With the exception of the circle dataset, where we used the uniform proposal, we used the cluster informed prior for both the PGSM and SAMS samplers in subsequent experiments.

The frequency of the resampling step had a larger effect.
A critical implementation point in order for the PGSM method to work is that resampling should be done adaptively by monitoring the ESS of the particle approximations in Algorithm \ref{alg:pgsm} \citep{Liu1995AdaptResampling, Lee2011}. 
We varied the relative ESS resampling threshold $\ressThreshold$ from 0 (never resample) to 1 (always resample) (Figures~\ref{fig:pilot} \textbf{e} and \textbf{f}).
We observed that the performance is markedly degraded if resampling is performed after each SMC iteration, but similar for all other resampling thresholds.
We used a threshold of 0.5 in all other experiments.

\begin{figure}
	\includegraphics{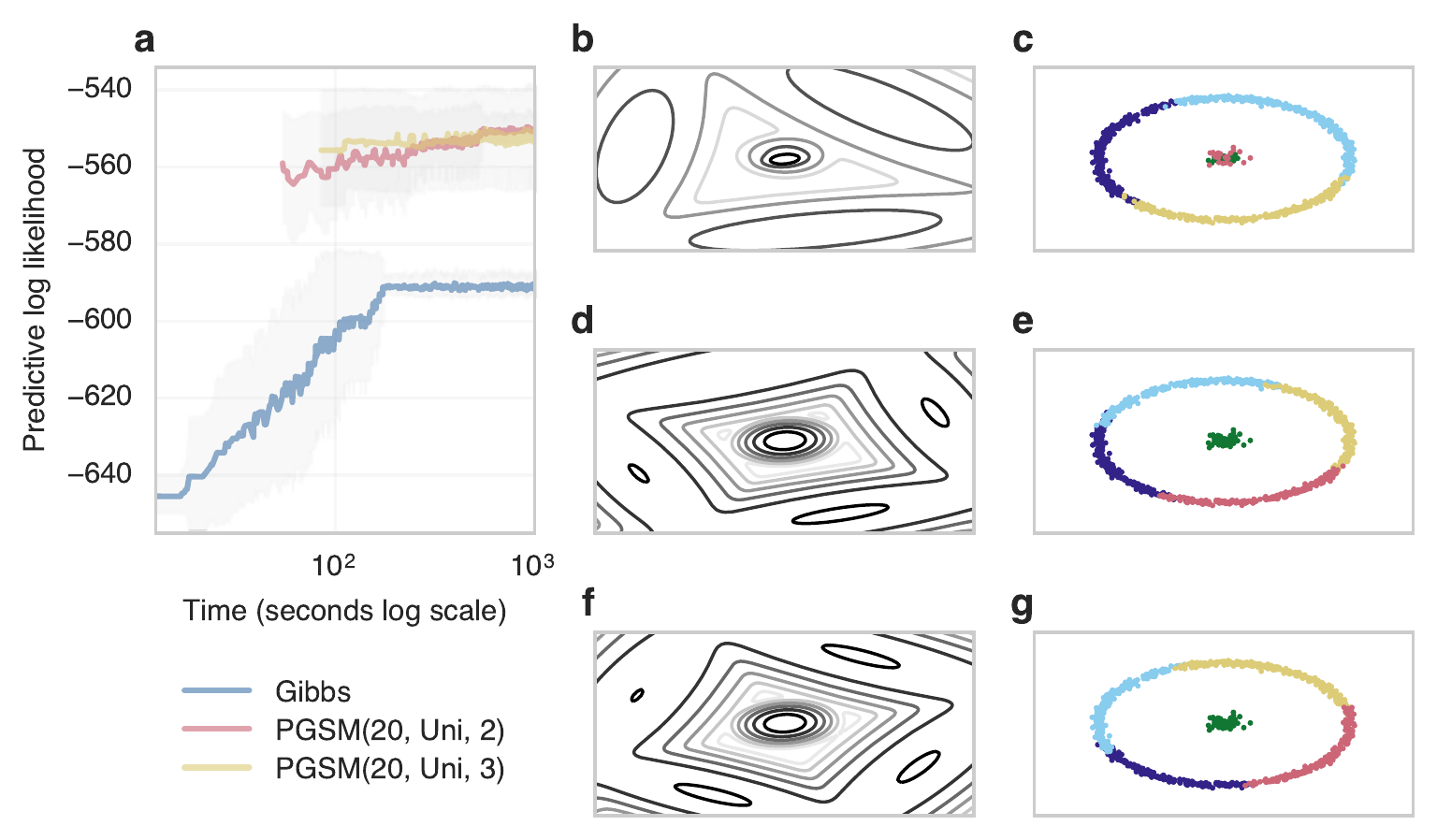}
	\caption{
		Comparison of Gibbs and PGSM for finite Dirichlet prior with $\maxnclust=5$.
		\textbf{a}) Predictive log likelihood comparison of Gibbs and PGSM using two (PGSM(20, Uni, 2)) or three (PGSM(20, Uni, 3)) anchors.
		Predictive density after 1000 seconds for \textbf{b}) Gibbs; \textbf{d}) PGSM(20, Uni, 2); \textbf{f}) PGSM(20, Uni, 3).
		Cluster assignment after 1000 seconds for \textbf{c}) Gibbs; \textbf{e}) PGSM(20, Uni, 2); \textbf{g}) PGSM(20, Uni, 3).
	}\label{fig:finite}
\end{figure}

Next, we used the dataset C1 to investigate the effectiveness of our method with the finite clustering model introduced in Section~\ref{sec:mixture}, with the number of clusters fixed to $\maxnclust=5$.
In this case, standard split-merge methods such as SAMS are less helpful since only merging can be performed when the maximum number of clusters has been allocated. 
The PGSM sampler does not have this restriction and naturally allows simultaneously splitting and merging while preserving the total number of clusters.
Furthermore, the PGSM sampler can use more than two anchors, potentially allowing for large changes in configuration without altering the number of clusters.
We compared the PGSM with two ($|s|=2$) and three ($|s|=3$) anchors to the Gibbs sampler.
The PGSM method outperformed the Gibbs sampler, though increasing the number of anchors did not improve the performance (Figure~\ref{fig:finite} \textbf{a}).
We plot the predictive densities (Figure~\ref{fig:finite} \textbf{b}, \textbf{d}, \textbf{f}) and cluster allocations (Figure~\ref{fig:finite} \textbf{c}, \textbf{e}, \textbf{g}) after running each sampler for 1000 seconds.
At this point the PGSM sampler used a single cluster to model the points in the middle, while the Gibbs samplers used two clusters to model the central cluster.

\begin{figure}[h!]
	\includegraphics{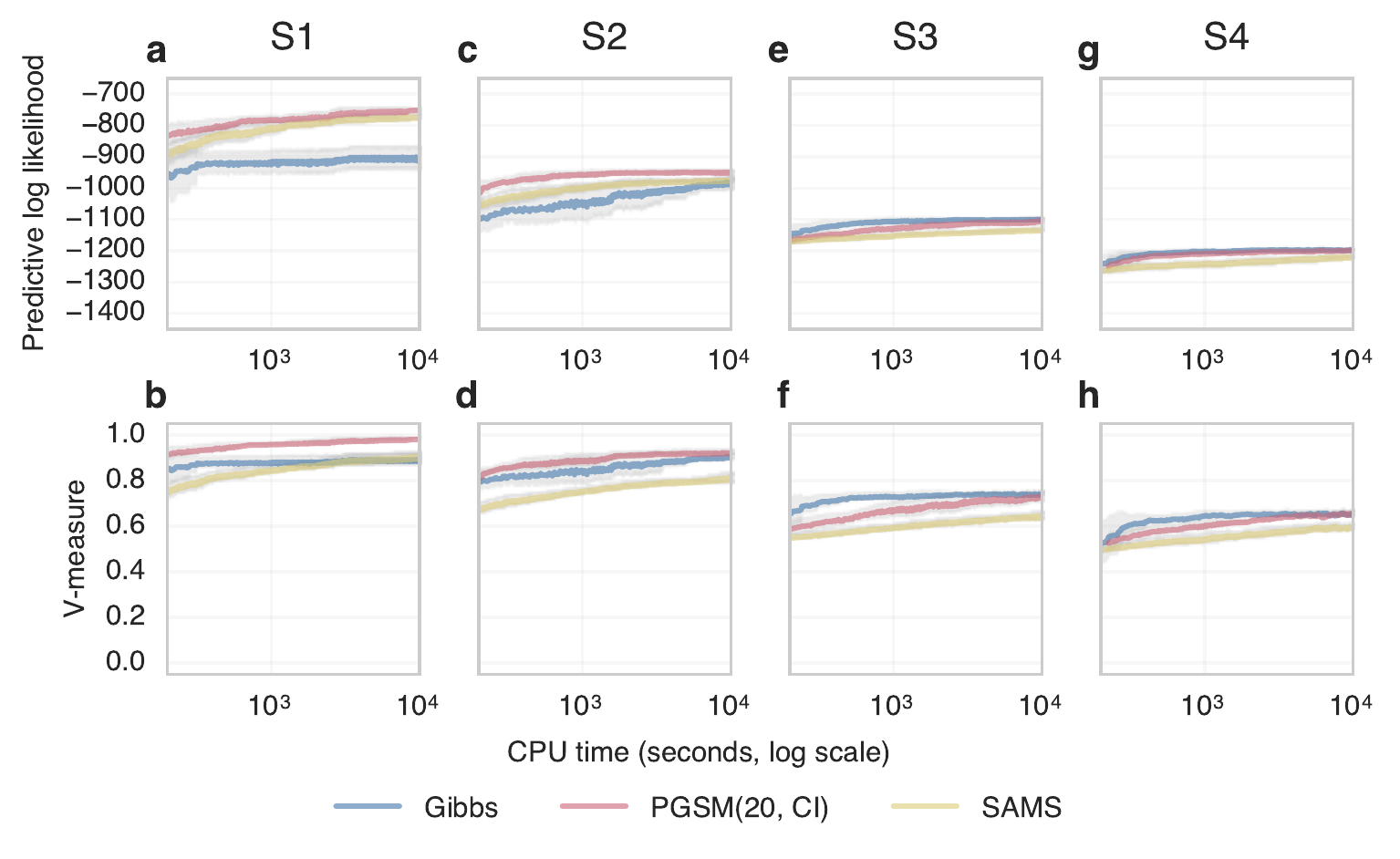}
	\caption{
		Comparison of MCMC algorithm using only a single kernel at a time (pure kernels) on 2D Normal datasets.
		Predictive log likelihood for datasets \textbf{a}) S1; \textbf{c}) S2; \textbf{e}) S3; \textbf{g}) S4.
		V-measure for datasets \textbf{b}) S1; \textbf{d}) S2; \textbf{f}) S3; \textbf{h}) S4.
	}
	\label{fig:main-synthetic-pure}
\end{figure}

\begin{figure}[h!]
	\includegraphics{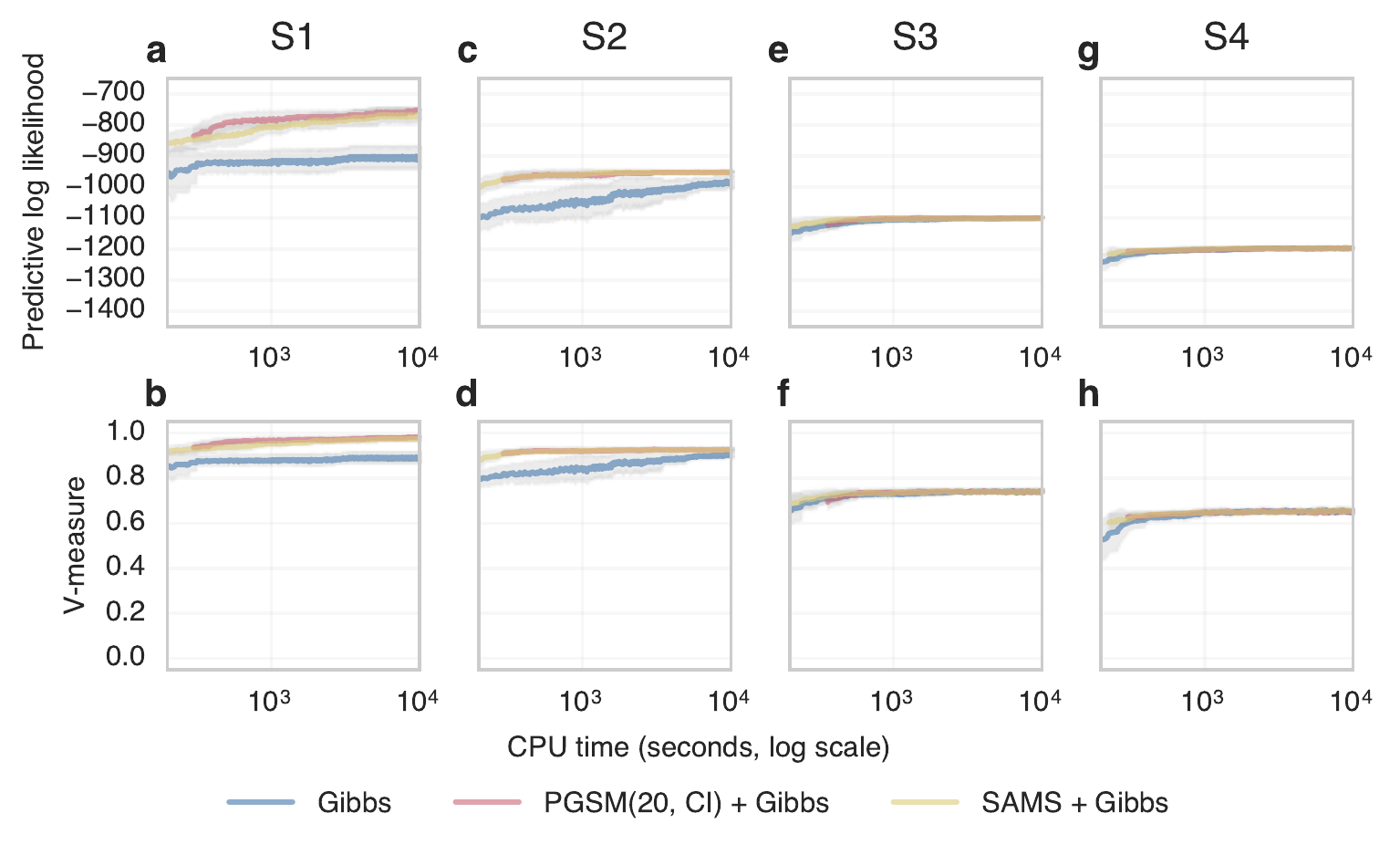}
	\caption{
		Comparison of MCMC algorithm using split-merge moves combined with standard Gibbs moves (mixed kernels) on 2D Normal datasets.
		Predictive log likelihood for datasets \textbf{a}) S1; \textbf{c}) S2; \textbf{e}) S3; \textbf{g}) S4.
		V-measure for datasets \textbf{b}) S1; \textbf{d}) S2; \textbf{f}) S3; \textbf{h}) S4.
	}
	\label{fig:main-synthetic-mixed}
\end{figure}

In Figures~\ref{fig:main-synthetic-pure} and \ref{fig:main-synthetic-mixed} we show a series of experiments on the four datasets S1--4 describe in the previous section.
We compare the PGSM to standard Gibbs and the SAMS method of \cite{dahl2005}.
We first compared pure kernels, where the split-merge samplers are not mixed with standard Gibbs moves (Figure~\ref{fig:main-synthetic-pure}).
The pure PGSM kernel outperformed both Gibbs and SAMS on datasets S1 and S2.
The Gibbs kernel and PGSM perform similarly for datasets S3 and S4, and both outperformed SAMS.
When the split-merge moves are mixed with standard Gibbs moves, the split-merge methods	 outperformed Gibbs on datasets S1 and S2, with all methods showing similar performance on datasets S3 and S4 (Figure~\ref{fig:main-synthetic-mixed}).

\begin{figure}[h!]
	\centering
	\includegraphics{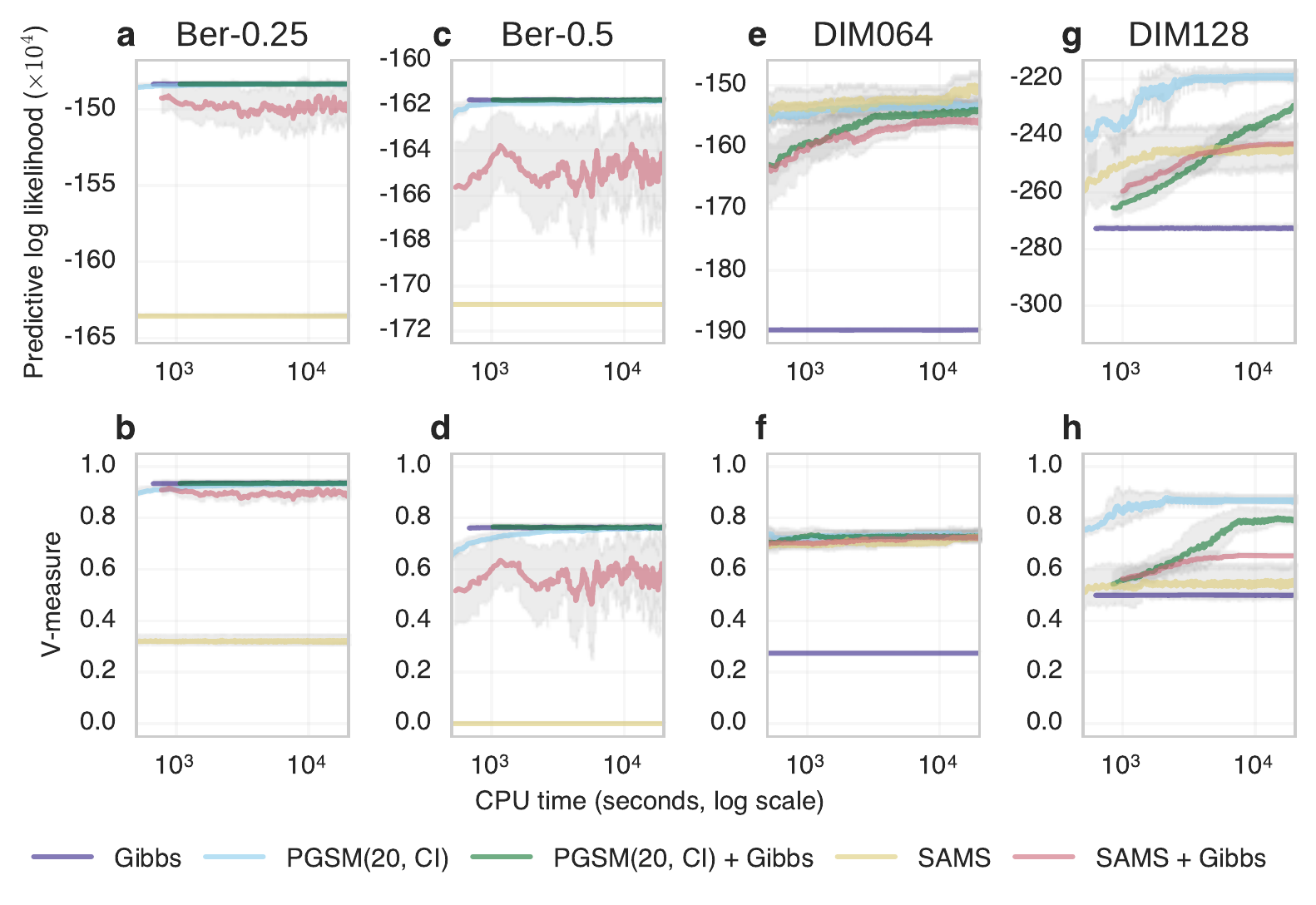}
	\caption{
		Comparison of MCMC algorithm using split-merge moves with combined with standard Gibbs moves (mixed kernels) on 50 dimensional Bernoulli data with 25\% (Ber-0.25) and 50\% (Ber-0.5) and Normal data  64 (DIM064) and 128 (DIM128) dimensions.
		Predictive log likelihood for datasets \textbf{a}) Ber-0.25; \textbf{c}) Ber-0.5; \textbf{e}) DIM064; \textbf{g}) DIM128.
		V-measure for datasets \textbf{b}) Ber-0.25; \textbf{d}) Ber-0.5; \textbf{f}) DIM064; \textbf{h}) DIM128.
	}
	\label{fig:main-synthetic-high-dim}
\end{figure}

Finally, we explored the performance of the methods on four high dimensional datasets.
The mixed PGSM and Gibbs samplers performed the best on the Bernoulli datasets, while the unmixed PGSM sampler is slower to reach the same predictive likelihood and V-measure (Figures~\ref{fig:main-synthetic-high-dim} \textbf{a}-\textbf{d}).
The mixed SAMS sampler failed to reach the same predictive likelihood as the PGSM and Gibbs methods, oscillating around lower values.
The unmixed SAMS sampler appears to be trapped in a local mode, corresponding to poor predictive likelihood and V-measure.
For the Normal datasets, the Gibbs sampler was trapped in a local mode and had markedly worse performance than other methods (Figures~\ref{fig:main-synthetic-high-dim} \textbf{e}-\textbf{h}).
The unmixed samplers outperformed the mixed equivalents on the 64 dimensional data.
Furthermore, the unmixed PGSM method had a large performance advantage over all other methods on the 128 dimensional dataset.

\subsection{Geolocation data}\label{sec:geo}

We compared the performance of the three sampling methods on a geolocation dataset.
The dataset, described in more detail in \cite{Franti2010Mopsi}, consists of a subset of data collected by MOPSI, a Finnish mobile application where users can post their current geographic location via their mobile device.
The subset we used consists of a list of 13,467 locations (latitude-longitude pairs) from users located in Finland until 2012.
The data is freely accessible from \url{http://cs.joensuu.fi/sipu/datasets/}.

We use this data as a proxy for the estimation of mobile device user density.
The DP mixture of Normal-Inverse-Wishart distributions provides a natural way to obtain a parsimonious estimate of population density, where the flexibility on the shape and number of clusters can accommodate a broad range of density variability factors ranging from densely populated cities to vast low-density rural areas.

We summarize the results in Figure~\ref{fig:mopsi}.
In Figures~\ref{fig:mopsi} \textbf{a}-\textbf{c}, we display quantitative results as measured by held-out predictive likelihood performance.
In Figure~\ref{fig:mopsi} \textbf{a}, we show that mixed PGSM, mixed SAMS and Gibbs samplers perform similarly.
In Figure~\ref{fig:mopsi} \textbf{b}, we show that the performance of SAMS is considerably degraded if SAMS is not mixed with a Gibbs kernel.
In Figure~\ref{fig:mopsi} \textbf{c}, we show that the performance of PGSM is less degraded if not mixed with a Gibbs kernel.

In Figures~\ref{fig:mopsi} \textbf{e}-\textbf{j}, we visualize the posterior predictive density approximated using MCMC samples.
We also show the raw data in Figure~\ref{fig:mopsi} \textbf{d} for reference.
The following three pairs of density plots are included to illustrate the high computational cost of initializing a standard Gibbs sampler at a fully disconnected configuration.
From left to right: the first pair shows the predictive density after one round of Gibbs sampling initialized at the fully disconnected configuration (Figures~\textbf{e} and \textbf{f}); the second pair, after sampling with PGSM initialized at the fully connected configuration for the same time (Figures~\textbf{g} and \textbf{h}); the third, after running the Gibbs sampler for $10^5$ seconds (Figures~\textbf{i} and \textbf{j}).
This demonstrates that our method can produce accurate and compact density estimates without relying on an expensive initialization phase.

\begin{figure}[h!]
	\includegraphics[scale=0.95]{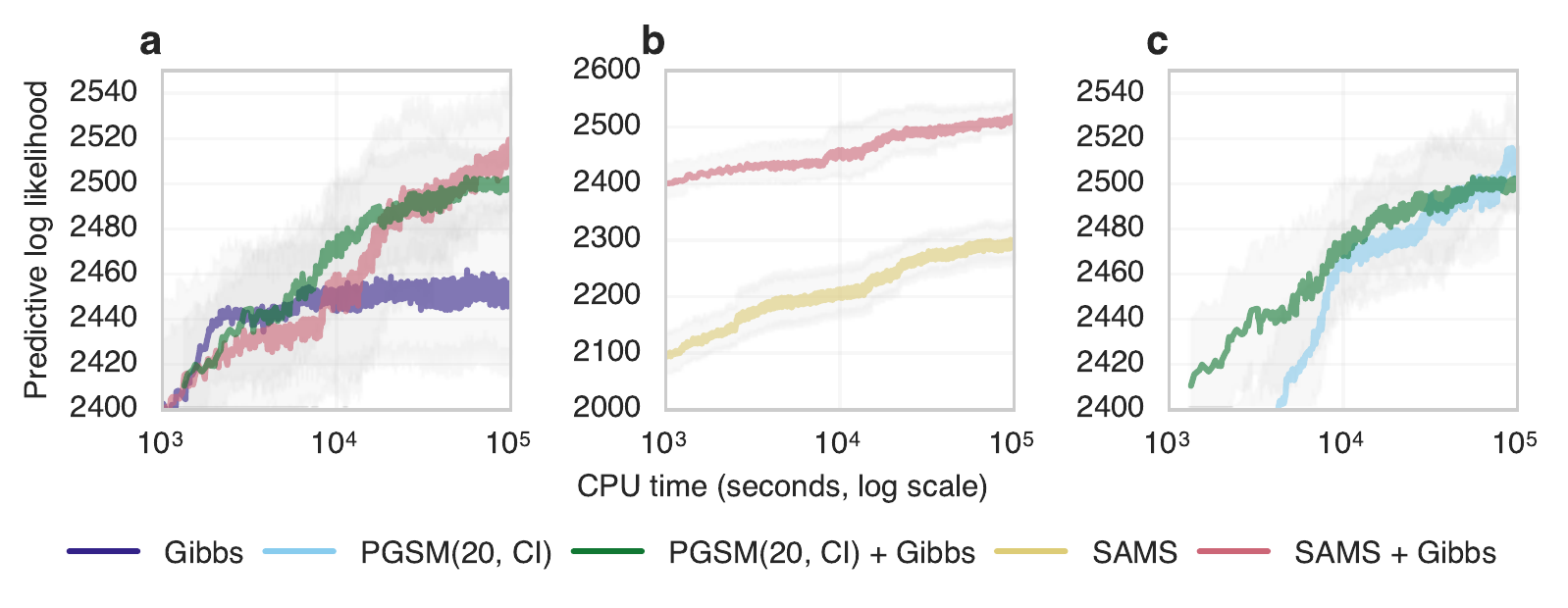}
	\includegraphics[scale=0.95]{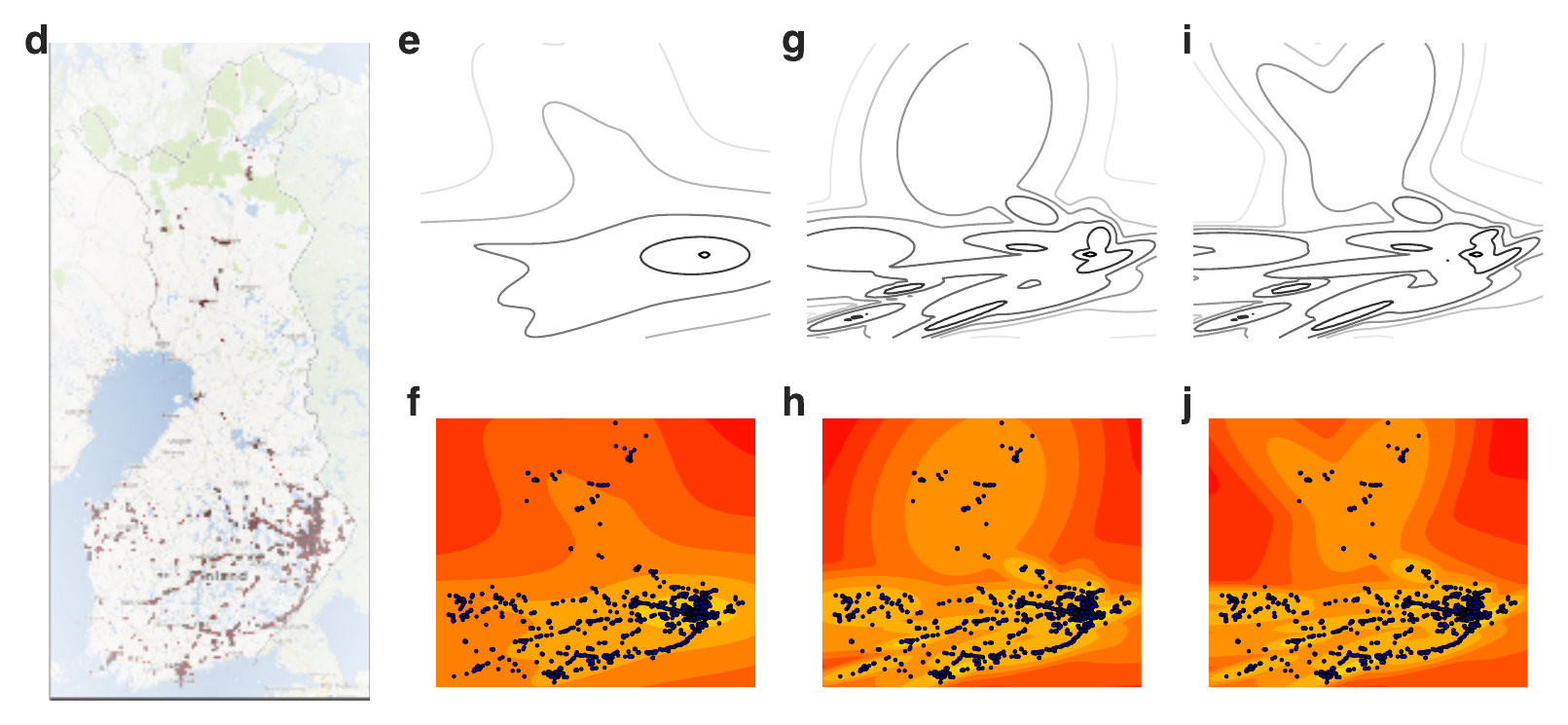}
	\caption{
		Geolocation dataset.
		Comparison of predictive likelihoods of \textbf{a}) mixed kernels; \textbf{b}) PGSM pure and mixed kernels; \textbf{c}) SAMS pure and mixed kernels.
		\textbf{d}) MOPSI geolocation dataset.
		\textbf{e}) Predictive density of Gibbs initialized from the disconnected configurations after one iteration;  \textbf{d}) with point data points overlayed.
		\textbf{g}) and \textbf{h}) Predictive density of mixed PGSM kernel after an equivalent amount of time.
		\textbf{i}) and \textbf{j}) Predictive density of disconnected Gibbs after $10^5$ seconds.
	}
	\label{fig:mopsi}
\end{figure}

\subsection{Inferring population structure in heterogeneous tumours}\label{sec:genomics}

The PyClone model \citep{roth2014pyclone} is designed to infer the proportion of cancer cells in a tumour sample which contain a mutation, which we refer to as the cellular prevalence of the mutation.
The input data consists of a set of digital measurements of allelic abundance which is assumed to be proportional to the true abundance of the allele in the sample.
The key factors which need to be deconvolved to convert this measurement to an estimate of cellular prevalence are that some cells derive from healthy (normal) tissue and the genomes of cancer cells contain multiple copies of a locus.
The model assumes that mutations will group by cellular prevalence due to the expansion of populations of genetically identical cells.
The number of populations is unknown, thus the PyClone model uses a DP prior with a Uniform($\left[0, 1\right]$) base measure.
The component parameters are interpreted as the cellular prevalence of the mutations associated with the component.

We show results on a dataset with 10,000 synthetic mutations in Figure~\ref{fig:cancer-results}. 
All methods except the pure SAMS kernel performed similarly in terms of predictive likelihood, while the pure SAMS kernel performed significantly worse (Figure~\ref{fig:cancer-results} \textbf{a}).
The pure PGSM and SAMS kernels outperformed the other methods in terms of V-measure, though the difference were small (Figures~\ref{fig:cancer-results} \textbf{b}).

As observed in the other domains, the performance of SAMS critically depends on mixing the kernel with GIBBS moves.
We show the data points for each replication of the pure split-merge kernels further supporting this point (Figure~\ref{fig:cancer-results} \textbf{c}).

\begin{figure}[t]
	\includegraphics{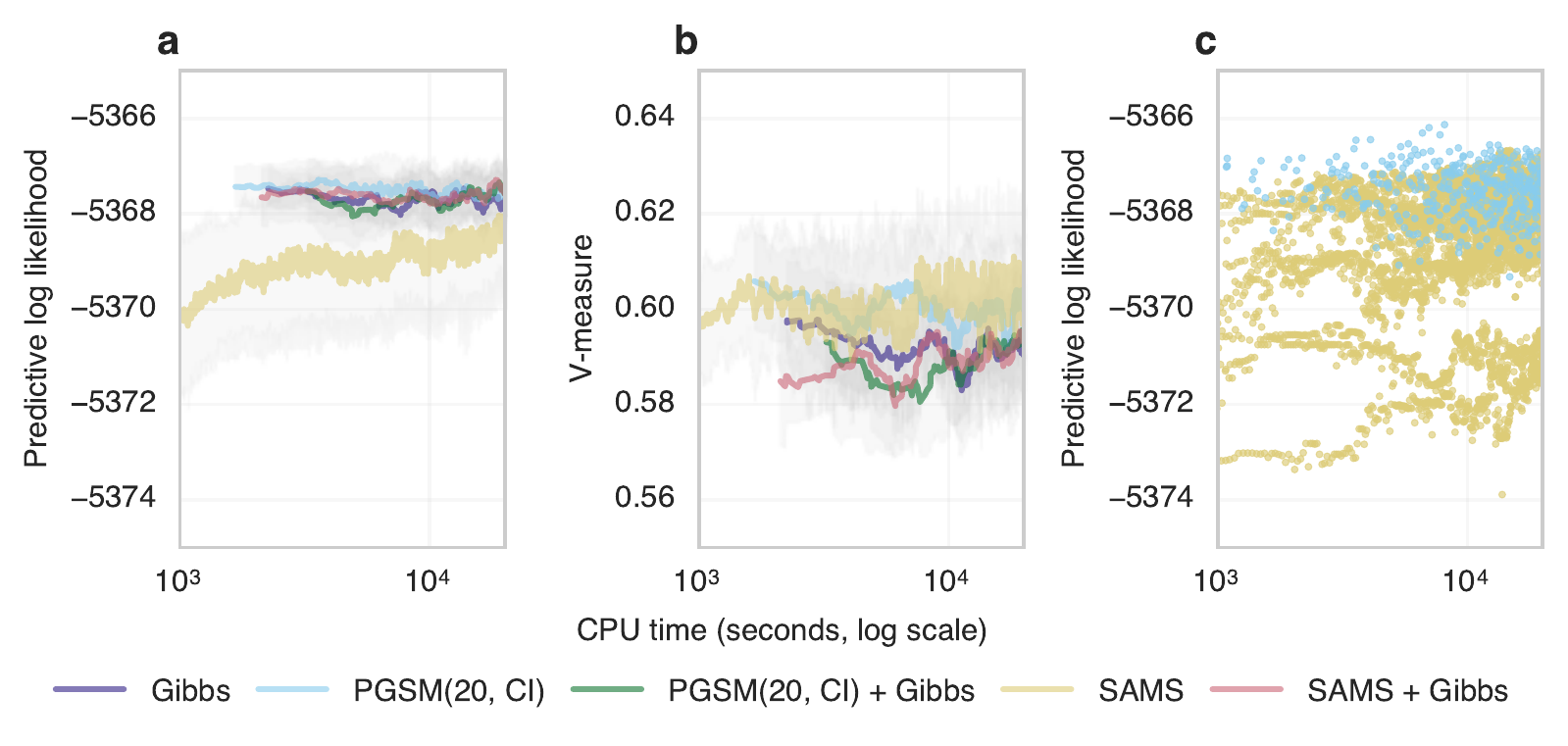}
	\label{fig:cancer-results}
	\caption{
		Clustering of cancer mutations from synthetic next-generation sequencing data.
		\textbf{a}) Comparison of predictive likelihood.
		\textbf{b}) Trace of log predictive likelihood for PGSM and SAMS pure kernels from 10 random restarts.
	}
\end{figure}

\section{Discussion}\label{sec:discussion}

We have proposed a new methodology to design efficient split-merge moves for Bayesian mixture models.
The method also generalizes to new types of moves useful for finite clustering models when $|s|>2$.
We have shown empirically that the proposed method is competitive in a range of clustering and likelihood models, including synthetic and real datasets from geolocation and genomics applications.

Our method, being based on the established PMCMC framework, opens up many directions for future improvements.
This includes applying recent advances in parallel implementations of SMC, for example via graphical processing units \citep{Lee2010GraphicCards}, or modifications of the SMC algorithm itself \citep{Jun2012Entangled,Murray2012GPU,Lee2014Forest}.

Another area of improvement comes from the development of resampling schemes tailored to discrete latent variables. In Algorithm \ref{alg:pgsm}, the number of possible distinct successors for each given particle is a small finite number (at most two if $|s|=2$ for example).
The complexity of the problem comes from the fact that a potentially long sequence of such decisions need to be made in order to split a cluster.
In these specific scenarios, custom PMCMC methods based on the early work of \cite{Fearnhead2003DiscretePF} have been developed in \cite{Whiteley2010Tech} and would provide futher improvement.

The fact that the state transitions $\partSupport(\cdot)$ have an absorbing state has both advantages and disadvantages. On the one hand it may cause Algorithm \ref{alg:pgsm} to be greedy, as explained in Section~\ref{sec:improved}. We have described in the same section a choice of intermediate and proposal distributions tailored to alleviate this issue. A potential alternative consists in designing a resampling distribution $\r$, which conditions on the survival of at least one representative of both a merge and a split. None of the existing resampling schemes have this property.
On the other hand, having an absorbing state has the advantage that if all particles simulated by Algorithm \ref{alg:pgsm} at some iteration $\t$ are equal to the merge absorbing state (i.e. $\x_\t^\p = \#2$ for all particle index $\p \in \{1, \dots, \N\}$), then there is no need to continue the computation of the particle filter for $\t' > \t$.

In standard applications of the PG algorithm, coalescence of the particle genealogy may cause slow mixing as noted in \cite{andrieu2010}.
The issue is that the particles $\xVec_\n^1, \xVec_\n^2, \dots, \xVec_\n^\N$ appearing in Algorithm \ref{alg:pgsm} have components at time $\t$ for $\t \ll \n$ which coincide with high probability with the components of the conditioning path.
This can be resolved using more sophisticated MCMC moves on the PG auxiliary variables \citep{Whiteley2010Discussion,Whiteley2010Tech,Lindsten2014PGAS}.
In our non-standard setup, this issue is partially mitigated by the fact that the order $\sigmaVec$ at which the particles are introduced is itself random.
Nonetheless, it would be interesting to implement these more advanced schemes to the problem at hand.

We have shown in Section~\ref{sec:genomics} a simple and effective method for handling models where each cluster component is governed by a non-conjugate model with a low-dimensional parameterization.
We leave for future work the extension of our method to higher dimensional non-conjugate likelihood models.
This problem can be approached, for example, by combining our method with the auxiliary variables described in \cite{neal2000}.

\section*{Acknowledgements}

Alexandre Bouchard-C\^{o}t\'{e}'s research was funded by an NSERC Discovery Grant.
Arnaud Doucet's research is partially supported by EPSRC grants EP/K000276/1 and EP/K009850/1.
Andrew Roth was partially supported by a CIHR CGS scholarship.
Computing was supported by WestGrid and Compute Canada.

\appendix

\section{Correctness of the decomposition into split-merge subproblems}\label{appendix:split-merge-correctness}

We present in this section the proof of correctness of the decomposition of the clustering into split-merge sub-problems.
The main tool used to prove this result is an auxiliary variable construction.
The auxiliary variable consists of a pair $(\s, \cMinus)$, where $\s$ is the set of anchors, and $\cMinus$ consists of the blocks of the partition $\c$ that do not contain anchor points:
\begin{eqnarray}
\cMinus \defeq \{\b\in\c : \b \cap \s = \emptyset\},
\end{eqnarray}
These intuitively correspond to the blocks of the partition that are forced to stay unchanged in this split-merge step.
We will view the split-merge step as a Gibbs step conditioning on $\cMinus, \s$.

A slight subtlety is that conditioning on the auxiliary variables not only forces the blocks in $\cMinus$ to stay constant; it also forces the other blocks to each contain at least one of the anchors.
See Figure~\ref{fig:support-example} for an example.
This leads to condition~\ref{assumption:2fromlemma} in Lemma~\ref{lemma:main}.

\begin{figure}[t]
	\begin{center}
		\includegraphics[width=5in]{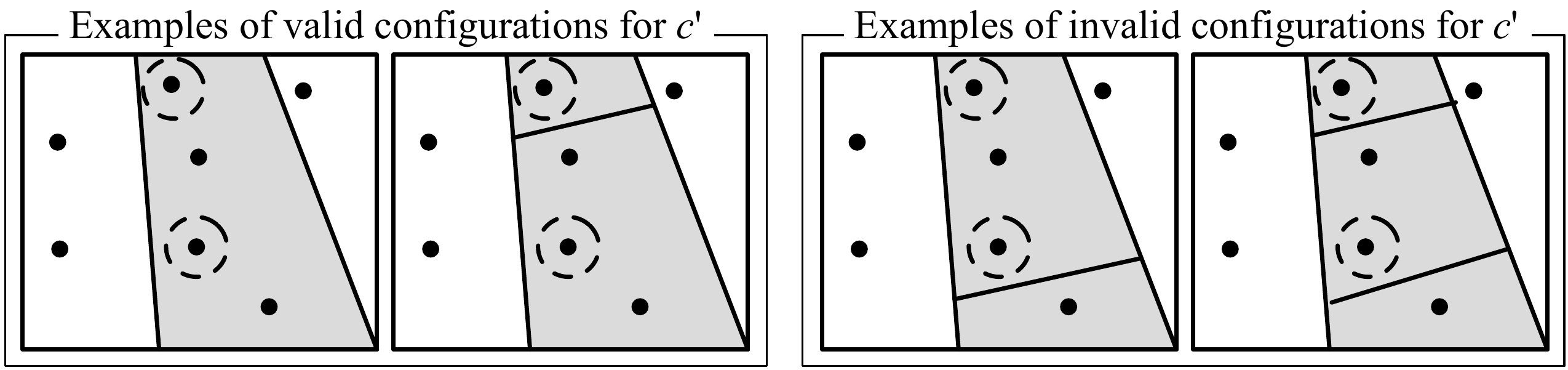}
		\caption{
			Assuming that the values of $\c, \cMinus$ and $\s$ are as in Figure~\ref{fig:example}, this illustrates some examples of configurations $\c'$ that are part of (left), or excluded from (right), the support of $\c'|\s, \cMinus$.
		}
		\label{fig:support-example}
	\end{center}
\end{figure}

\begin{lemma}\label{lemma:main}
	Let $\c, \c'$ denote two partitions of $[\T]$. Let $\s \subseteq [\T]$.
	Define $\cMinus \defeq \{\b \in \c : \s \cap \b = \emptyset\}$ and $\cMinus' \defeq \{\b \in \c' : \s \cap \b = \emptyset\}$.
	Then $\cMinus = \cMinus'$ if and only if the following two conditions hold:
	\begin{enumerate}
		\item $\c \cap \cMinus = \c' \cap \cMinus$, and, \label{assumption:1fromlemma}
		\item $\b \in \c' \backslash \cMinus \Longrightarrow \b \cap \s \neq \emptyset$.  \label{assumption:2fromlemma}
	\end{enumerate}
\end{lemma}

\begin{proof}
	$(\Longrightarrow)$ Condition~\ref{assumption:1fromlemma} holds trivially.
	For condition~\ref{assumption:2fromlemma}, suppose (a) $\cMinus = \cMinus'$, (b), $\b \in \c' \backslash \cMinus$, but (c) $\b \cap \s = \emptyset$.
	By (b), $\b\in\c'$ and $\b\notin\cMinus$. This and (c) implies that $\b\in\cMinus'$.
	But this contradicts (a), so condition~\ref{assumption:2fromlemma} holds as well.
	
	$(\Longleftarrow)$ First, suppose $\b\in\cMinus$.
	By condition~\ref{assumption:1fromlemma}, $\b\in\cMinus'$.
	Therefore, $\cMinus \subseteq \cMinus'$.
	
	Second, suppose $\b\in\cMinus'$.
	By the contrapositive of condition~\ref{assumption:2fromlemma}, $\b \notin \c'\backslash \cMinus$.
	This point and $\cMinus' \subseteq \c'$ implies that $\b\in\cMinus$.
\end{proof}

We can now turn to the proof of Proposition~\ref{prop:split-merge-correctness}.
We copy its statement here for convenience:

\begin{proposition}
	If $\c \sim \pi$, then the output of Algorithm~\ref{alg:setup_split_merge}, $\c'$, satisfies $\c' \sim \pi$; i.e. the Markov kernel $\K(\c'|\c)$ induced by Algorithm~\ref{alg:setup_split_merge} is $\pi$-invariant.
\end{proposition}

\begin{proof}
	Consider the model augmented with the auxiliary variables $\s$ and $\cMinus$ (see Figure~\ref{fig:graphical-models}(a)), defined formally using the following auxiliary distribution:
	\begin{equation}
	\piTilde(\s, \cMinus, \c) \defeq \pi(\c) \h(\s) \1[\cMinus = \cMinus(\s,\c)],
	\end{equation}%
	where $	\cMinus(\s, \c) \defeq \{\b\in\c : \b \cap \s = \emptyset\}$. Note that this auxiliary distribution admits the target distribution as a marginal:
	\begin{eqnarray}
	\sum_\s \sum_{\cMinus} \piTilde(\s, \cMinus, \c) &=& \pi(\c) \sum_\s \h(\s) \sum_{\cMinus}\1[\cMinus = \cMinus(\s,\c)]  \\
	&=& \pi(\c) \sum_\s \h(\s) = \pi(\c), \nonumber
	\end{eqnarray}
	where the sum over $\cMinus$ is over all sets of subsets of $[\T]$, and the sum over $\s$ is over all subsets of $[\T]$.
	We used the fact that only one $\cMinus$ satisfies $\cMinus(\s,\c) = \cMinus$, and that $\h$ is a probability mass function.

	Next, we introduce three kernels with inputs and outputs denoted by:
	\begin{equation}
	\c \labelledmapsto{\K_1} (\s, \cMinus, \c) \labelledmapsto{\K_2} (\s, \cMinus, \c') \labelledmapsto{\K_3} \c'.
	\end{equation}
	
	These kernels play the following roles:
	\begin{itemize}
		\item $\K_1$ samples the auxiliary variables according to $\piTilde(\s, \cMinus\mid\c)$, while keeping $\c$ fixed,
		\item $\K_2$ performs a Metropolis-within-Gibbs step on $\c$ targeting the auxiliary distribution $\piTilde$,
		\item $\K_3$ deterministically projects the triplet back to the original space, retaining only the clustering $\c$.
	\end{itemize}
	
	Formally:
	\begin{eqnarray}
	\K_1(\s', \cMinus', \c'\mid \c) &\defeq& h(\s') \1[\c = \c'] \1[\cMinus' = \cMinus(\s, \c)], \\
	\K_2(\s', \cMinus', \c'\mid \s, \cMinus, \c) &\defeq& \piTilde(\c' | \s, \cMinus) \1[\s' = \s] \1[\cMinus' = \cMinus], \nonumber \\
	\K_3(\c'|\s, \cMinus, \c) &\defeq& \1[\c' = \c]. \nonumber
	\end{eqnarray}%
	Since $\piTilde$ admits $\pi$ as a marginal, the composition of $\K_1, \K_2,$ and $\K_3$ is clearly $\pi$-invariant.
	It is therefore enough to show that when $\cMinus = \cMinus(\s, \c)$ where $\c$ is a valid partition of $[\T]$, sampling from $\K_2$ is equivalent to sampling from the Markov kernel $\K(\c'|\c)$ induced by Algorithm~\ref{alg:setup_split_merge}:
	\begin{eqnarray}\label{eq:equivalence-last}
	\piTilde(\c'\mid \s, \cMinus) &\propto& \pi(\c') \1[\cMinus = \cMinus(\s, \c')] \\
	&=&  \pi(\c') \1[\cMinus(\s, \c) = \cMinus(\s, \c')] \nonumber \\
	&=& \tauI(|\c'|) \left( \prod_{\b\in\c'} \tauII(|\b|) \L(\yVec_\b) \right) \1[\cMinus(\s, \c) = \cMinus(\s, \c')]. \nonumber
	\end{eqnarray}
	
	Using Lemma~\ref{lemma:main}, we now rewrite the support as follows:
	\begin{eqnarray}
	\1[\cMinus(\s, \c') = \cMinus(\s, \c)] &=& \1[\c \cap \cMinus = \c' \cap \cMinus] \1[\b \in \c' \backslash \cMinus \Longrightarrow \b \cap \s \neq \emptyset].
	\end{eqnarray}
	Let now $\cBar' = \c' \backslash \cMinus$. Plugging in the last line of Equation~(\ref{eq:equivalence-last}), we obtain:
	\begin{eqnarray}\label{eq:last}
	\piTilde(\c'\mid \s, \cMinus) &=& \tauIBar(|\cBar'|) \left( \prod_{\b\in\cBar'} \tauII(|\b|) \L(\yVec_\b) \right) \1[\c \cap \cMinus = \c' \cap \cMinus] \1[\b \cap \s \neq \emptyset]   \\
	&=& \piBar(\cBar') \1[\c \cap \cMinus = \c' \cap \cMinus], \nonumber
	\end{eqnarray}
	where $\piBar(\cBar')$ is defined in Equation (\ref{eq:piBar}). Since Algorithm~\ref{alg:setup_split_merge} does not change the clustering of points outside of $\sBar$ (line \ref{step:create-final-cluster} of Algorithm~\ref{alg:setup_split_merge}), it follows that the indicator 
	function in the last line of Equation~(\ref{eq:last}) is equal to one.
\end{proof}

\section{Correctness of particle Gibbs for split merge}\label{appendix:pg-correctness}

We provide here the proof of Proposition~\ref{prop:pg-correctness}.
The main steps in the proof follow a structure similar to the proof of Proposition~\ref{prop:split-merge-correctness}.

\begin{proposition}
	Under Assumption~\ref{assumption:support}, \ref{assumption:bijection}, and \ref{assumption:intermediate-final}, and if $\cBar \sim \piBar$, then the output of Algorithm~\ref{alg:pgsm}, $\cBar'$, satisfies $\cBar' \sim \piBar$ for any $N\geq2$; i.e. the Markov kernel $\KBar(\cBar'|\cBar)$ induced by Algorithm~\ref{alg:pgsm} is $\piBar$-invariant.
\end{proposition}

\begin{proof}
	We augment the model $\cBar$ with the auxiliary variables $\sigmaVec$ and $\g$ (see Figure~\ref{fig:graphical-models}(b)), defined as:
	\begin{enumerate}
		\item $\sigmaVec$ is distributed according to the output of Algorithm~\ref{alg:sample_permmutation}, defined in Section~\ref{sec:pgsm-overview}.
		\item Given $\sigmaVec$ and $\cBar$, the variables $\g = (\ancestors_{2:\n}, \xVec_{1:\n}^{1:\N}, k)$ are distributed according to the specification of Algorithm~\ref{alg:pgsm}, with the exception that all particle indices are shuffled according to an independent permutation of $\{1, \dots, \N\}$ at each generation.
		Here $k$ is the index of the particle sampled at iteration $\n$ (on line (\ref{alg:line:sample_particlepath})).
	\end{enumerate}
	
	\begin{figure}[t]
		\begin{center}
			\includegraphics[width=3in]{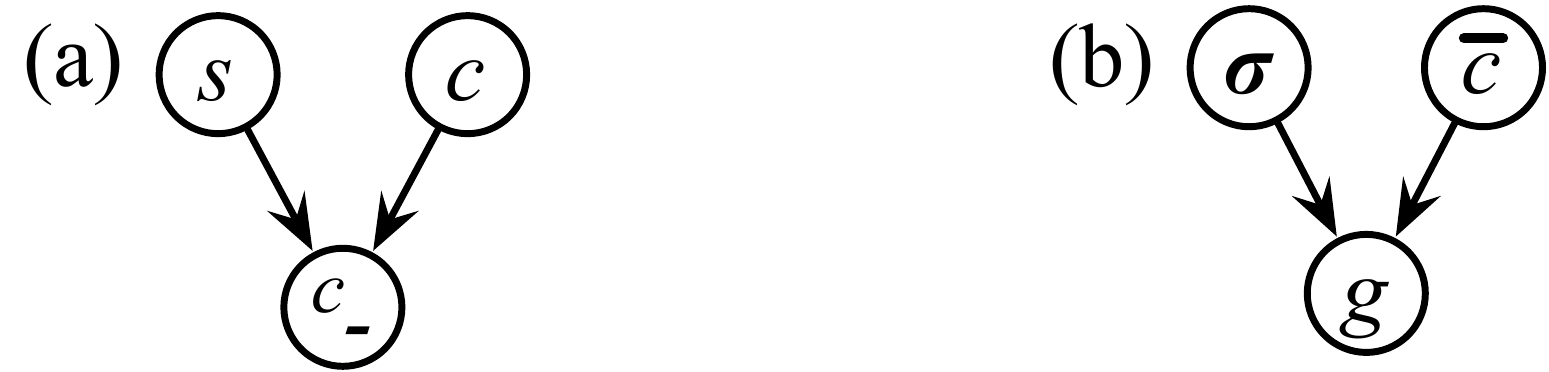}
			\caption{
				Graphical models of the auxiliary variables used in the correctness proofs.
				The structure of the dependencies give an intuitive justification that the original model can be recovered as a marginal in both cases, as there are not directed path from the auxiliary variables to the original variables.
				(a) In Appendix~\ref{appendix:split-merge-correctness}, the auxiliary variables are $\s$ and $\cMinus$, and the original variable is $\c$.
				(b) In Appendix~\ref{appendix:pg-correctness}, the auxiliary variables are $\sigmaVec$ and $\g$, and the original variable is $\cBar$.
			}
			\label{fig:graphical-models}
		\end{center}
	\end{figure}
	
	Next, we introduce three kernels with inputs and outputs denoted by:
	\begin{equation}
	\cBar \labelledmapsto{\KBar_1} (\sigmaVec, \cBar) \labelledmapsto{\KBar_2} (\sigmaVec, \g, \cBar') \labelledmapsto{\KBar_3} \cBar'.
	\end{equation}
	
	These kernels play the following roles:
	\begin{itemize}
		\item $\KBar_1$ samples the permutation $\sigmaVec$ while keeping the auxiliary variables $\cBar$ fixed,
		\item $\KBar_2$ samples $g$ using the PG step then sets $\cBar'$ to $\phi^\sigmaVec(\xVec^k_\n)$,
		\item $\KBar_3$ deterministically projects the triplet back to the original space, retaining only the restricted clustering $\cBar'$.
	\end{itemize}
	
	The kernel $\KBar_2$ is equivalent to a standard PG algorithm.
	Assumption~\ref{assumption:support}, \ref{assumption:intermediate-final}, and Theorem~5(a) of \cite{andrieu2010} imply that $\KBar_2$ is $\piBar$-invariant (and in fact, irreducible).
	Assumption~\ref{assumption:bijection} ensures that the computation of the conditioned path is well-defined.
\end{proof}

\section{Construction of the bijection}\label{appendix:bijection}

We provide here the proof of Proposition~\ref{prop:bijection}:

\begin{proposition}
	For any permutation $\sigmaVec$ satisfying $\{\sigma_1, \sigma_2\} = \s$, there is a bijective map $\phi^\sigmaVec$ from the space of particles respecting the transition constraints, $\partSupport_\n$, to the support of the restricted target, $\support(\piBar)$.
\end{proposition}

\begin{proof}
	Consider the following mapping:
	\begin{equation}
	\phi^\sigmaVec(\xVec_\t) \defeq \bracearraycond{
		\{\{\sigma_1, \dots, \sigma_\t\}\}&\;\;\textrm{if }\x_\t \in \{\#1, \#2\}, \\
		\{\overline{\sigma}_1(\xVec_\t), \overline{\sigma}_2(\xVec_\t)\}&\;\;\textrm{otherwise,}
	}
	\end{equation}
	where $\overline{\sigma}_i(\xVec_\t) \defeq \{\sigma_{\t'} : \x_{\t'} = \#(2+i), 1 \le \t' \le \t\}$. It is easy to check that it has an inverse given by:
	\begin{equation}
	\left(\left(\phi^\sigmaVec\right)^{-1}(\cBar)\right)_\t \defeq \bracearraycond{
		\#1&\;\;\textrm{if }\t = 1, \\
		\#2&\;\;\textrm{if }\t > 1, |\cBar| = 1, \\
		\#3&\;\;\textrm{if }\t > 1, |\cBar| > 1, \sigma_1 \isclusteredwith_{\cBar} \sigma_\t, \\
		\#4&\;\;\textrm{if }\t > 1, |\cBar| > 1, \sigma_2 \isclusteredwith_{\cBar} \sigma_\t,
	}
	\end{equation}
	where $\sigma_i \isclusteredwith_{\cBar} \sigma_j$ means that $\y_{\sigma_i}$ is in the same block as $\y_{\sigma_j}$ for the clustering $\cBar$.
	By the construction of the support of $\piBar$, exactly one of the four cases above holds when $\cBar \in \support(\piBar)$.
	
\end{proof}

\section{Generalization to $|\s| > 2$}\label{appendix:generalization}

We describe here the algorithmic implications of increasing the number of anchor points, $|\s|$, to some  constant greater than two.
This constant should be selected so that the number of partitions of $|\s|$ points is much lower than the number of particles.

The algorithm is generally unchanged, with the following exceptions:

\begin{enumerate}
	\item Algorithm~\ref{alg:sample_permmutation} is modified to sample $(\sigma_1, \sigma_2, \dots, \sigma_{|\s|})$ uniformly over the permutations of $\s$, and $(\sigma_{|\s|+1}, \dots, \sigma_\n)$, over the permutations of $\sBar \backslash \s$,
	\item as before, the local allocation state space $\X$ can be viewed as a pair each containing a partition and a block in this partition (see Figure~\ref{fig:local-state-space}).
	In the case where $|\s| = 2$, the partitions are taken from the union of the set of partitions of a set of size one with the set of partitions of a set of size two. When $|\s| > 2$, we add more states, corresponding to partitions of a set of size three, etc. until we add states corresponding to partitions of a set of size $|\s|$.
	The support of the transition $\partSupport$ consists in (a) edges $\x \to \x'$ linking a state $\x'$ such that removing one element from one of its blocks yields $\x$, and (b) edges $\x \to \x'$ where $\x$ and $\x'$ correspond to the same partition of a set of size $|\s|$.
	This is a generalization of the case $|\s| = 2$ shown in Figure~\ref{fig:local-state-space}.
	The mapping $\phi^\sigmaVec$ is generalized in the obvious way,
	\item in Section~\ref{sec:improved} the following equations are substituted,
	\begin{enumerate}
		
		\item $\t \in \{1, 2\} \rightarrow \t \in \{1, 2, \dots, |\s|\}$,
		\item $\t = 2 \rightarrow \t \in \{2, \dots, |\s|\}$,
		\item $\t > 2 \rightarrow \t > |\s|$,
		\item $\Delta \schedule \defeq (n - 2)^{-1} \rightarrow \Delta \schedule \defeq (n - |\s|)^{-1}$.
	\end{enumerate}
\end{enumerate}

\section{Models}\label{appendix:models}

\subsection{Multivariate normal}

The first likelihood we use is the multivariate normal (MVN) with density denoted $\mathcal{N} (y| \mu, \Sigma)$.
We specify a normal inverse Wishart (NIW) prior for the mean and covariance parameters with density denoted $\mathcal{N}I\mathcal{W} (\mu, \Sigma | \nu, r, u, S)$.
The densities are given by
\begin{eqnarray}
\mathcal{N}I\mathcal{W} (\mu, \Sigma | \nu, r, u, S) & = & \mathcal{N} \left( \mu |u, \frac{1}{r} \Sigma \right) \mathcal{I}\mathcal{W} (\Sigma | \nu, S),  \\
\mathcal{N} (y| \mu, \Sigma) & = & \frac{1}{(2 \pi)^{\frac{D}{2}} | \Sigma |^{\frac{1}{2}}} \exp \left( - \frac{1}{2} (y - \mu)^T \Sigma^{- 1}  (y - \mu) \right), \nonumber \\
\mathcal{I}\mathcal{W} (\Sigma | \nu, S) & = & \frac{| S |^{\frac{\nu}{2}}}{2^{\frac{\nu p}{2}} \Gamma_D \left( \frac{\nu}{2} \right)}  | \Sigma |^{- \frac{\nu + p + 1}{2}} \exp \left( - \frac{1}{2} \text{tr} (S \Sigma^{- 1}) \right), \nonumber
\end{eqnarray}
where $\Gamma_D (x) = \pi^{\frac{D (D - 1)}{4}}  \prod_{d = 1}^D \Gamma \left(x + \frac{d - 1}{2} \right)$.

We use the following priors for all experiments $(\nu, r, u, S) = (\nu_0, r_0, u_0, S_0) = (2 + D, 1, \tmmathbf{0}, \tmmathbf{I})$, where $\tmmathbf{0}$ is the $D$ dimensional vector of zeros, and $\tmmathbf{I}$ is the $D$ dimensional identity matrix.
The posterior distribution of $\mu, \Sigma$ given $\tmmathbf{y}= (y_1, \ldots, y_m)$ is $\mathcal{N}\mathcal{I}W (\mu, \Sigma |\nu_m, r_m, u_m, S_m)$ where
\begin{eqnarray}
\nu_m & = & \nu_0 + m,  \\
r_m & = & r_0 + m, \nonumber\\
u_m & = & \frac{r_0 u_0 + \sum_{i = 1}^m y_i}{r_m}, \nonumber \\
S_m & = & S_0 + \sum_{i = 1}^m y_i y_i^T + r_0 u_0 u_0^T - r_m u_m u_m^T. \nonumber
\end{eqnarray}
For computational efficiency it is convenient to express these updates iteratively using the following equations:
\begin{eqnarray}
\nu_m & = & \nu_{m - 1} + 1, \\
r_m & = & r_{m - 1} + 1, \nonumber\\
u_m & = & \frac{r_{m - 1} u_{m - 1} + y_m}{r_m}, \nonumber\\
S_m & = & S_{m - 1} + \frac{r_m}{r_{m - 1}}  (y_m - u_m)  (y_m - u_m)^T. \nonumber
\end{eqnarray}
Using these equations the Cholesky decomposition of $S_0$ can be performed once using \ $O (D^3)$ operations and cached.
This decomposition can then be updated using $m$ rank one updates, each requiring $O (D^2)$ operations, to obtain $S_m$.
This allows for efficient evaluation of the marginal and predictive likelihoods as $| S_m |$ can be evaluated using $O (D)$ operations using the Cholesky decomposition, instead of the standard $O (D^3)$ operations.

The marginal likelihood for the MVN-NIW congugate pair is
\begin{eqnarray}
L (\tmmathbf{y}) & = & \int \prod_{i = 1}^m L (y_i | \theta) H (\mathd \theta)  \\
& = & \int \prod_{i = 1}^m \mathcal{N} (y_i | \mu, \Sigma) \mathcal{N}I\mathcal{W} (\mu, \Sigma | \nu, r, u, S) \ud\mu \ud\Sigma \nonumber \\
& = & \frac{1}{\pi^{\frac{m D}{2}}} \frac{r_0^{\frac{D}{2}}}{r_m^{\frac{D}{2}}}  \frac{| S_0 |^{\frac{\nu_0}{2}}}{| S_m |^{\frac{\nu_m}{2}}}  \frac{\prod_{d = 1}^D \Gamma \left( \frac{\nu_m + d - 1}{2} \right)}{\prod_{d = 1}^D \Gamma \left(\frac{\nu_0 + d - 1}{2} \right)}. \nonumber
\end{eqnarray}
The predictive likelihood is given by
\begin{eqnarray}
L (\tmmathbf{y}^+ |\tmmathbf{y}^-) & = & \frac{L (y_1, \ldots, y_m)}{L(y_1, \ldots, y_{m - 1})}  \\
& = & \frac{1}{\pi^{\frac{D}{2}}}  \frac{r_{m - 1}^{\frac{D}{2}}}{r_m^{\frac{D}{2}}}  \frac{| S_{m - 1} |^{\frac{\nu_{m - 1}}{2}}}{| S_m |^{\frac{\nu_m}{2}}}  \frac{\prod_{d = 1}^D \Gamma \left(\frac{\nu_m + d - 1}{2} \right)}{\prod_{d = 1}^D \Gamma \left( \frac{\nu_{m - 1} + d - 1}{2} \right)}. \nonumber
\end{eqnarray}

\subsection{Bernoulli}

We use a Bernoulli likelihood, $\text{Bernoulli} (x| \theta)$, with a Beta prior distribution, $\text{Beta} (\theta | \alpha, \beta)$.
We use the following priors $(\alpha, \beta) = (\alpha_0, \beta_0) = (1, 1)$ for all experiments.
The densities are
\begin{eqnarray}
\text{Bernoulli} (x| \theta) & = & \theta^x (1 - \theta)^{1 - x},  \\
\text{Beta} (\theta | \alpha, \beta) & = & \frac{\Gamma (\alpha) \Gamma(\beta)}{\Gamma (\alpha + \beta)} \theta^{\alpha - 1} \theta^{\beta - 1}. \nonumber
\end{eqnarray}
The posterior density of $\theta$ given $\tmmathbf{y}= (y_1, \ldots, y_m)$ is $\text{Beta} (\alpha_m, \beta_m)$ where $\alpha_m = \alpha_0 + \sum_{i = 1}^m y_i$ and $\beta_m = \beta_0 + \sum_{i = 1}^m (1 - y_i)$.
The marginal likelihood is
\begin{eqnarray}
L (\tmmathbf{y}) & = & \int \prod_{i = 1}^m L (y_i | \theta) H (\mathd \theta)  \\
& = & \int \prod_{i = 1}^m \text{Bernoulli} (y_i | \theta)  \text{Beta}(\theta | \alpha_0, \beta_0) \ud\theta \nonumber \\
& = & \frac{\Gamma (\alpha) \Gamma (\beta)}{\Gamma (\alpha_m) \Gamma(\beta_m)}  \frac{\Gamma (\alpha_m + \beta_m)}{\Gamma (\alpha_0 + \beta_0)}, \nonumber
\end{eqnarray}
and the predictive log likelihood is
\begin{eqnarray}
L (\tmmathbf{y}^+ |\tmmathbf{y}^-) & = & \frac{L (y_1, \ldots, y_m)}{L(y_1, \ldots, y_{m - 1})}  \\
& = & \frac{\Gamma (\alpha_{m - 1}) \Gamma (\beta_{m - 1})}{\Gamma(\alpha_m) \Gamma (\beta_m)}  \frac{\Gamma (\alpha_m + \beta_m)}{\Gamma(\alpha_{m - 1} + \beta_{m - 1})}. \nonumber
\end{eqnarray}

\subsection{PyClone}

For the cancer genomics data we use the application-specific PyClone likelihood model over clonal prevalences, genotypes, and observed read counts.
The key variables in the model are as follows (see \cite{roth2014pyclone} for a more detailed description of the model):
\begin{eqnarray}
\phi_i &:& \text{proportion of cancer cells with mutation $i$, } \phi_i\in[0,1], \nonumber \\
t & : & \text{proportion of cancer cells in a sample (treated as known), } t \in  [0, 1], \nonumber \\
\psi_i & : & \text{genotype of normal, non-mutated cancer and mutated cancer cells, } \nonumber \psi_i \in (g_N, g_R, g_V), \nonumber \\
g_x & \in & \mathcal{G} = \left\{ \text{A}, \text{B}, \text{AA}, \text{AB}, \ldots \right\}, \nonumber \\
\pi_{i, \psi_i} & : & \text{probability that mutation $i$ has genotype $\psi_i$ (elicited from auxillary data)}, \nonumber \\
c (g_x) & = & \# \text{A} (g_x) +\# \text{B} (g_x), \nonumber \\
\mu (g_x) & = & \frac{\# \text{A} (g_x)}{c (g_x)}, \nonumber \\
\xi (\psi, \phi, t) & : & \text{probability of sampling a B from the population of cells in the sample, i.e.:} \nonumber \\
& = & \frac{(1 - t) c (g_N) \mu (g_N) + t (1 - \phi) c (g_R) \mu (g_R) + t \phi c (g_V) \mu (g_V) }{(1 - t) c (g_N) + t (1 - \phi) c (g_R) + t \phi c  (g_V)}, \nonumber \\
y_i & : & \text{number of sequence reads with a B and total number of reads covering mutation $i$, i.e.:} \nonumber \\
& = & (y_{i, b}, y_{i, d}) \in \mathbb{N}^2. \nonumber
\end{eqnarray}
The generative model is specified as follow:
\begin{eqnarray}
H_0 & = & \text{Uniform} ([0, 1]),  \\
\concentration & \sim & \text{Gamma} (\concentration |a, b), \nonumber \\
H| \concentration, H_0 & \sim & \text{DP} (H| \concentration, H_0), \nonumber \\
\phi_i |H & \sim & H, \nonumber \\
y_i | \psi_i, \phi_i, t & \sim &  \text{Binomial} (y_{i, b} |y_{i, d}, \xi(\psi_i, \phi, t)). \nonumber
\end{eqnarray}
This model is not conjugate. However, if we let $x \in \{ x_0, \ldots, x_M \} = \left\{ 0, \frac{1}{M - 1}, \ldots, \frac{M - 2}{M - 1}, 1 \right\}$ be a discretization of the interval $[0, 1]$ and replace the continuous uniform base measure, $H_0 = \text{Uniform} ([0, 1])$, with the discrete uniform measure, $H_0 = \text{Uniform} (\{ x_0, \ldots, x_M \})$, then we can approximate the model.
Using this approximation, we can now treat the model as if it were conjugate.
The marginal likelihood for data $(y_1,...,y_m)$ is given by
\begin{eqnarray}
\int \prod_{i = 1}^m L (y_i | \theta) H (\mathd \theta) & = & \int \prod_{i = 1}^m \sum_{\psi_i \in \mathcal{G}^3} \pi_{i, \psi_i}  \text{Binomial}(y_{i, b} |y_{i, d}, \xi (\psi_i, \phi_{}, t)) H (\mathd \phi_{})  \\
& = & \sum_{k = 0}^M \prod_{i = 1}^m \sum_{\psi_i \in \mathcal{G}^3} \pi_{i, \psi_i}  \text{Binomial} (y_{i, b} |y_{i, d}, \xi (\psi_i, x_k, t)) \frac{1}{M} \nonumber \\
& = & \sum_{k = 0}^M \prod_{i = 1}^m \exp \left( \underbrace{\log \sum_{\psi_i \in \mathcal{G}^3} \pi_{i, \psi_i}  \text{Binomial} (y_{i, b} |y_{i, d}, \xi (\psi_i, x_k, t)) }_{\Xi_k (y_i)} \right) \frac{1}{M} \nonumber \\
& = & \sum_{k = 0}^M \exp \left( \sum_{i = 1}^m \Xi_k (y_i) \right) \frac{1}{M}, \nonumber
\end{eqnarray}
where we have the sufficient statistics
\begin{eqnarray}
\tmmathbf{\Xi} (y_i) & = & (\Xi_0 (y_i), \ldots, \Xi_M (y_i)).
\end{eqnarray}
\begin{remark}
	The possibly infinite sum $\sum_{\psi_i \in \mathcal{G}^3}$ is truncated to a finite sum over biologically plausible states.
\end{remark}

\

\section{Anchor proposal distribution}\label{appendix:proposals}

The anchor proposal distribution, $h$, is a free tuning parameter for the PGSM sampler. 
In principle, proposals which are informed by the current clustering state of the chain or by the topology of the space may improve the performance of the sampler.

We consider two informed proposal distributions.
While bespoke proposals for each model may perform better, we restrict attention here to proposals which can be applied generically to any class of model for which the PGSM sampler is applicable.
In particular, we do not assume a distance metric is available.
Both proposals we discuss are only applicable when two anchor points are used.

\begin{algorithm}
	\caption{Cluster informed (CI) proposal}\label{alg:ci_proposal}
	\begin{algorithmic}[1]
		\State $i_{1} \sim \mbox{Uniform}([\T])$
		
		\State $\bar{b} \gets b \in c$ s.t. $i_{1} \in b$
		
		\State $c' \gets c \setminus \{\bar{b}\}$
		
		\For{$b \in c'$}
		\State $s_{b} \gets \frac{L(y_{\bar{b} \cup b})}{L(y_{\bar{b}}) L(y_{b})}$
		\EndFor
		
		\State $s_{\bar{b}} \gets \frac{\sum_{b \in c'} s_{b}}{|c| - 1}$ \Comment{Merge probability is set to $\frac{1}{|c| - 1}$}
		
		\For{$b \in c$}
		\State $p_{b} \gets \frac{s_{b}}{\sum_{b \in c} s_{b}}$ 
		\EndFor
		
		\State $b' \sim \mbox{Discrete}(c, p_{b})$ \Comment{Sample a block $b'$ in $c$ with probability $p_{b}$}
		
		\State $i_{2} \sim \mbox{Uniform}(b' \setminus {i_{1}})$
		
		\State \textbf{return} $i_{1}, i_{2}$
	\end{algorithmic}
\end{algorithm}

\begin{algorithm}
	\caption{Threshold informed (TI) proposal}\label{alg:ti_proposal}
	\begin{algorithmic}[1]
		\State $i_{1} \sim \mbox{Uniform}([\T])$
		
		\For{$b \in c$}
		\If{$i_1 \in b$}
		\State $b \gets b \setminus i_1$
		\EndIf
		\State $s_{b} \gets \tau_{2}(b) L(y_{i_{1}}|b)$ \Comment{CRP attachment probability where $L(\cdot|b)$ is the predictive distribution}
		\EndFor
		
		\For{$b \in c$}
		\State $p_{b} \gets \frac{s_{b}}{\sum_{b \in c} s_{b}}$ 
		\EndFor
		
		\State $b' \sim \mbox{Uniform}(\{b : p_{b} \ge t\})$ \Comment{$t$ is a pre-specified threshold, set to 0.01 in the experiments}
		
		\State $i_{2} \sim \mbox{Uniform}(b' \setminus {i_{1}})$
		
		\State \textbf{return} $i_{1}, i_{2}$
	\end{algorithmic}
\end{algorithm}

\begin{remark}
	If the any of the sets that we sample uniformly from are empty, we return two anchors sampled uniformly at random.
\end{remark}

\newpage
\renewcommand{\nomname}{List of Symbols}
{\footnotesize
\printnomenclature
}

\newpage
\bibliography{tex/poset_smc}

\end{document}